  \documentclass[nocopyrightspace,numbers]{sigplanconf}

\title{Mathematical Execution: \\  A Unified Approach for Testing Numerical Code}

\authorinfo{Zhoulai Fu\qquad\qquad\and\and ~~Zhendong Su}
  {Department of Computer Science, University of California, Davis, USA}
  {zhoulai.fu@gmail.com \and\qquad\and  su@cs.ucdavis.edu}

\makeatletter
\newcommand{\removelatexerror}{\let\@latex@error\@gobble}
\makeatother
\usepackage{todonotes}

\usepackage[all]{xy}

\usepackage{siunitx}
\usepackage{epigraph}

\usepackage[scaled=0.81]{beramono} 
\usepackage{enumitem}
\usepackage{graphicx}
\usepackage{xcolor}
\usepackage{amsmath}
\usepackage{amsfonts}
\usepackage{amssymb}
\usepackage{amsthm}
\usepackage[hyphens]{url} 

\usepackage{multicol}
\usepackage{multirow}
\usepackage{xspace}

\usepackage{ifthen}

\usepackage{float} 
\usepackage{microtype}
\usepackage{mathptmx}
\usepackage{stmaryrd} 
\usepackage{textcomp}

\usepackage{fixltx2e} 
\usepackage{booktabs}
\usepackage{colortbl} 
\usepackage{tabularx}
\usepackage{fancyhdr}
\usepackage[mathletters]{ucs}
\usepackage[utf8x]{inputenc}
\usepackage{wasysym}	
\usepackage[normalem]{ulem}	
\usepackage{balance} 
\usepackage[font=small]{caption}
\usepackage{subcaption}
\usepackage[T1]{fontenc}
\usepackage{fixltx2e} 
\usepackage{framed}
\usepackage[linesnumbered,ruled,vlined]{algorithm2e}
\SetCommentSty{mycommfont}

\usepackage[%
	bookmarks, colorlinks, citecolor=blue, urlcolor=black, 
	filecolor=black, linkcolor=blue  
]{hyperref}

\usepackage{listings}
\usepackage{bbm}
    \usepackage[bottom]{footmisc}
\usepackage{bbold}

\newcommand{\vs}{\hbox{\emph{vs.}}\xspace}

\newcommand{\etal}{\hbox{\emph{et al.}}\xspace}
\newcommand{\eg}{\hbox{\emph{e.g.}}\xspace}
\newcommand{\ie}{\hbox{\emph{i.e.}}\xspace}
\newcommand{\st}{\hbox{\emph{s.t.}}\xspace}

\newcommand{\etc}{\hbox{\emph{etc.}}\xspace}
\newcommand{\aka}{\hbox{\emph{a.k.a.}}\xspace} 

\newcommand{\mydef}{\overset{\text{def}}{=}}

\newtheorem{thm}{Theorem}[section]
\newtheorem{lem}[thm]{Lemma}
\newtheorem{cor}[thm]{Corollary}

\theoremstyle{definition}
 
\newtheorem{definition}[thm]{Definition}
\newtheorem{example}[thm]{Example}
 \newtheorem{remark}[thm]{Remark}

\newcommand{\Real}{\mathbbm{R}}

\newcommand{\energy}{f}

\SetKwFunction{RF}{RF}
\SetKwFunction{proc}{Basin-Hopping}
\SetKwFunction{procshort}{MH}
\definecolor{light-gray}{gray}{0.8}
\newcommand{\highlight}[1]{\colorbox{light-gray}{#1}}
\newcommand{\dom}[1]{\operatorname{dom}(#1)\xspace}

\newcommand{\FOO}{{\texttt {FOO}}\xspace}

\newcommand{\FOOI}{{\texttt {FOO\_I}}\xspace}
\newcommand{\FOOR}{{\texttt {FOO\_R}}\xspace}
\newcommand{\loader}{{\texttt {loader}}\xspace}
\newcommand{\librso}{{\texttt {libr.so}}\xspace}
\newcommand{\mcmc}{{\texttt {MCMC}}\xspace}
\newcommand{\pen}{{\mathnormal{pen}}\xspace}
 \newcommand{\myr}{\texttt {r}\xspace}

\newcommand\sem[1]{\left[\!\left|#1\right|\!\right]}

\newcommand{\about}[1]{\paragraph{}\noindent{{\underline{\sc  #1}}}}
\renewcommand{\about}[1]{}

\newcommand{\niter}{\mathnormal {n\_iter}\xspace}
\newcommand{\nStart}{\mathnormal {n\_start}\xspace}
\newcommand{\LM}{\texttt {LM}\xspace}
\newcommand{\lmin}{\texttt {LM}\xspace}

\newcommand{\dist}{\mathnormal{d}}

\usepackage{tikz}

\newcommand{\Integer}{\mathbbm{Z}}

\newcommand{\err}{\mathbbm{E}}

\newcommand{\slug}{\hbox{\kern1.5pt\vrule width2.5pt height6pt depth1.5pt\kern1.5pt}}
\def\xskip{\hskip 7pt plus 3pt minus 4pt}
\newdimen\algindent
\newif\ifitempar \itempartrue 
\def\algindentset#1{\setbox0\hbox{{\bf #1.\kern.25em}}\algindent=\wd0\relax}
\def\algbegin #1 #2{\algindentset{#21}\alg #1 #2} 
\def\aalgbegin #1 #2{\algindentset{#211}\alg #1 #2} 
\def\alg#1(#2). {\medbreak 
  \noindent{\bf#1}({\it#2\/}).\xskip\ignorespaces}
\def\algstep#1.{\ifitempar\smallskip\noindent\else\itempartrue
  \hskip-\parindent\fi
  \hbox to\algindent{\bf\hfil #1.\kern.25em}%
  \hangindent=\algindent\hangafter=1\ignorespaces}

\newcommand{\ME}{\textnormal{ME}\xspace}
\newcommand{\coverme}{{\textnormal{CoverMe}}\xspace}

\newcommand{\fdlibm}{\textnormal{Fdlibm}}
\DeclareMathOperator{\Explored}{Saturate}

\newcommand{\myT}{_{\mathnormal T}}
\newcommand{\myF}{_{\mathnormal F}}
\newcommand{\mylhs}{\mathnormal{a}}
\newcommand{\myrhs}{\mathnormal{b}}
\newcommand{\myop}{\mathnormal{op}}
\newcommand{\myU}{\mathnormal{U}}
\newcommand{\myM}{\mathnormal{M}}
\newcommand{\myX}{\mathnormal{X}}

\newcommand{\repf}{\mathnormal{R}}

\newcommand{\failure}{\Downarrow {\texttt{wrong}}}

\newcommand{\Paragraph}[1]{\paragraph{\upshape #1}}

\newcommand{\xloc}{x_{\textnormal{L}}}
\newcommand{\xpro}{\widetilde{\xloc}}

\begin{document}

\maketitle

%



\begin{abstract}

This paper presents \emph{Mathematical Execution} (ME), a new, unified approach for testing numerical code.  The key idea is to (1) capture the desired testing objective via a
\emph{representing function} and (2) transform the automated testing
problem to the minimization problem of the representing function. The
minimization problem is to be solved via  \emph{mathematical
  optimization}.  The main feature of ME is that it
directs input space exploration by only executing the representing function,
thus avoiding static or symbolic reasoning about the program semantics, which is particularly challenging for numerical code. To illustrate this feature, we
develop an ME-based algorithm for coverage-based testing of numerical code.
We also show the potential of applying and adapting ME
to other related problems, including path reachability testing,
boundary value analysis, and satisfiability checking.

To demonstrate ME's practical benefits, we have implemented \coverme,
a proof-of-concept realization for branch coverage based testing, and
evaluated it on Sun's C math library (used in, for example, Android,
Matlab, Java and JavaScript). We have compared \coverme with random
testing and Austin, a publicly available branch coverage based testing
tool that supports numerical code (Austin combines symbolic execution
and search-based heuristics). Our experimental results show that
\coverme achieves near-optimal and substantially higher coverage
ratios than random testing on all tested programs, across all
evaluated coverage metrics.  Compared with Austin, \coverme improves
branch coverage from 43\% to 91\%, with significantly less time (6.9
\vs 6058.4 seconds on average).
\end{abstract}

\section{Introduction}
\label{sect:intro}

Testing has been a predominant approach for improving software quality.
Manual testing is notoriously tedious~\cite{Brooks:1995:MM:207583};
{automated testing} has been an active research topic,
drawing on a rich body of techniques, such as symbolic
execution~\cite{King:1976:SEP:360248.360252,Boyer:1975:SFS:800027.808445,Clarke:1976:SGT:1313320.1313532,DBLP:journals/cacm/CadarS13},
random
testing~\cite{Bird:1983:AGR:1662311.1662317,DBLP:conf/pldi/GodefroidKS05,AFL:web}
and search-based
strategies~\cite{Korel:1990:AST:101747.101755,McMinn:2004:SST:1077276.1077279,Baars:2011:SST:2190078.2190152,Lakhotia:2010:FSF:1928028.1928039}.

Automated testing is about producing program
failures~\cite{DBLP:journals/computer/Meyer08}. Let $\FOO$ be a
program and $\dom{\FOO}$ be its input domain.  An automated testing
problem is to systematically find $x\in\dom{\FOO}$ such that
$\FOO(x)\failure$, where $\FOO(x)\failure$ denotes ``\FOO goes
wrong if executed on input $x$.''  It is difficult to specify $\FOO(x)\failure$, which is
known as the test oracle problem~\cite{DBLP:journals/cj/Weyuker82,DBLP:conf/kbse/MemonBN03}.
This paper assumes that an algorithm for checking $\FOO(x)\failure$ is given.


An important problem in testing is the testing of numerical
  code, \ie, programs with floating-point arithmetic, non-linear
variable relations, or external function calls (such as logarithmic and
trigonometric functions).  These programs are pervasive in
safety-critical systems, but ensuring their quality remains difficult.
Numerical code presents two
specific challenges for existing automated testing techniques: (1)
Random testing is easy to employ and fast, but ineffective in finding
deep semantic issues and handling large input spaces; and (2) symbolic execution and its
variants can perform systematic path exploration, but suffer from
path explosion and are weak in dealing with complex program logic
involving numerical constraints.

\Paragraph{Our Approach.}

This paper introduces a new, unified approach for automatically
testing numerical code.  It proceeds as follows: We derive from the
program under test \FOO another program \FOOR, called
\emph{representing function}, which represents how far an input
$x\in\dom{\FOO}$ is from reaching the set
$\{x\mid \FOO(x)\failure\}$.  We require that the representing
function returns a non-negative value for all $x$, which diminishes when $x$
gets close to the set and vanishes when $x$ goes inside.  Intuitively,
this representing function is similar to a sort of distance. It allows
to approach the automated testing problem, \ie, the problem of
finding an element in $\{x\mid \FOO(x)\failure\}$, as the problem of
minimizing \FOOR. This approach can be justified with a strong
guarantee:
\begin{align}
\FOO(x) \failure  ~\Leftrightarrow~ x \text{ minimizes } \FOOR,
\label{eq:intro:me}
\end{align}
assuming that there exists at least one $x$ such that
$\FOO(x)\failure$ (details in Sect.~\ref{sect:theory}).  Therefore,
the essence of our approach is to transform the automated testing
problem to a minimization problem. Minimization
problems are well studied in the field of Mathematical Optimization
(MO)~\cite{Minoux86}.  MO works by executing its objective function
only (see Sect.~\ref{sect:background}). That is to say, our approach
does not need to analyze the semantics of the tested programs.
Instead, it directs input space exploration by only executing the
representing function. We call this approach \emph{Mathematical
  Execution} (abbreviated as $\ME$).



Note that mathematical optimization by itself does not
necessarily provide a panacea for automated testing because many MO
problems are themselves intractable. However, efficient algorithms
have been successfully applied to difficult mathematical optimization
problems.  A classic example is the NP-hard {traveling salesman
  problem}, which has been nicely handled by simulated annealing~\cite{Kirkpatrick83optimizationby}, a stochastic MO technique.
Another example is \emph{Monte Carlo Markov Chain}~\cite{DBLP:dblp_journals/ml/AndrieuFDJ03}, which has been effectively 
adapted to testing and
verification~\cite{Schkufza:2014:SOF:2594291.2594302,DBLP:conf/oopsla/FuBS15,heule2015mimic,xsat}.
A major finding of this work is that using mathematical
optimization for testing numerical code 
is a powerful approach. If we carefully design the
representing function so that certain conditions are respected, we
can come up with mathematical optimization problems that can be
efficiently solved by off-the-shelf MO tools.


To demonstrate the feasibility of our approach, we have applied $\ME$  on 
\emph{coverage-based testing}~\cite{DBLP:dblp_books/daglib/0012071} of floating-point code, a fundamental problem in testing. The experimental results show that our
implemented tool, \coverme, is highly effective.  Fig.~\ref{fig:intro:fdlibm}
gives a small program from our benchmark suite
$\fdlibm$~\cite{fdlibm:web}. The
program operates on two {\tt double} input parameters. It first takes {\tt
  |x|}'s high word by bit twiddling, including a bitwise {\tt AND}
($\&$), a pointer reference (\&) and a dereference (*) operator. The bit
twiddling result is stored in integer variable {\tt ix} (Line 3),
followed by four conditional statements that examine {\tt ix} (Lines
4--15).  The tool  \coverme yields:

\begin{center}
                      \begin{tabular}{ll}
                        $100\%$ & line coverage\\
                        $87.5\%$  & branch coverage
                      \end{tabular}
\end{center}

When investigating why \coverme fails to  achieve full
branch coverage, we find that one out of the eight branches in the program cannot be reached. The condition {\tt
if ((int) x) == 0} (Line 5) always holds because it is nested
within the $|x|<2^{-27}$ branch (Line 4).~\footnote{Sun's developers decided to
use this redundant check to trigger the \emph{inexact exception} of
floating-point as a side effect. From the program semantics
perspective, no input of {\tt \_\_kernel\_cos} can trigger the
false branch of {\tt if (((int) x) == 0)}.}  Therefore, the $87.5\%$ branch coverage
is, in fact, optimal. We have compared \coverme with Austin~\cite{lakhotia2013austin},
a publicly available, state-of-the-art coverage-based testing tool that can handle floating-point code.  Austin achieves 37.5\% branch coverage in 1885.1 seconds, whereas \coverme achieves the optimal coverage in 15.4 seconds (see Sect.~\ref{sect:eval}).


\begin{figure}
\lstset{xleftmargin=0.5cm, numbers=left}
\begin{lstlisting}
#define __HI(x) *(1+(int*)&x)
double __kernel_cos(double x, double y){
  ix = __HI(x)&0x7fffffff;  /* ix = |x|'s high word */
  if(ix<0x3e400000) {       /* if |x| < 2**(-27) */
     if(((int)x)==0) return ...; /* generate inexact */
  }
  ...;
  if(ix < 0x3FD33333) 	  /* if |x| < 0.3 */ 
    return ...;
  else {
    if(ix > 0x3fe90000) { /* if |x| > 0.78125 */
      ...;
    } else {
      ...;
    }
    return ...;
  }
}
\end{lstlisting}
\caption{The benchmark program {\tt \_\_kernel\_cos}  taken from the $\fdlibm$~\cite{fdlibm:web} library (\url{http://www.netlib.org/fdlibm/k\_cos.c}).} 
\label{fig:intro:fdlibm}
\end{figure}



 




\Paragraph{Contributions.} 
 Our contributions follow:
\begin{itemize}[nosep]
\item We introduce Mathematical Execution, a new general approach for
  testing numerical code;
\item We develop an effective coverage-based testing algorithm using
  the $\ME$ approach;
\item We demonstrate that \ME is a unified
  approach by showing how to
  apply $\ME$ to several important testing problems; and
\item We implement the coverage-based testing tool \coverme and show
  its effectiveness on real-world numerical library code.
\end{itemize}

\Paragraph{Paper Outline.}
Sect.~\ref{sect:background} gives the background on mathematical
optimization.  Sect.~\ref{sect:overview} illustrates $\ME$ by studying
the case of branch coverage based testing.  We define the problem,
demonstrate the $\ME$ solution, and give the algorithmic details.
Sect.~\ref{sect:theory} lays out the theoretical foundation for $\ME$ and
demonstrates $\ME$ with several additional
examples. Sect.~\ref{sect:eval} presents an implementation overview of $\coverme$
and describes our experimental results. Sect.~\ref{sect:discussion}
discusses the current limitations of $\ME$.
Finally, Sect.~\ref{sect:relwork} surveys related work and
Sect.~\ref{sect:conc} concludes.  For completeness,
Appendix~\ref{sect:untested} lists the benchmark programs in Fdlibm
that $\coverme$ does not support and their reasons, and 
Appendix~\ref{sect:implem} gives implementation details.

\Paragraph{Notation.}
The sets of real and integer numbers are denoted by $\Real$ and  ${\Integer}$ respectively.
For two real numbers $a$ and $b$, the usage a{\sc E}b  means  $a * 10^b$.  
In this presentation, we do not distinguish a mathematical expression,
such as $x^2 + |y|$, and its implementation, such as \lstinline!x*x + abs(y)!. Similarly,
 we use a
lambda expression to mean either  a mathematical function or its
implementation. For example, an implementation $\lambda
x. x^2$ may refer to  the code  \lstinline!double f (double x) {return x*x;}!.
 We use the C-like syntax $A?~ v_1: v_2$ to mean an implementation that
  returns $v_1$ if $A$ holds, or $v_2$ otherwise.

\section{Background}
\label{sect:background}
We begin with some preliminaries on mathematical optimization
following the exposition of~\cite{xsat}. A complete treatment of
either is beyond the scope of this
paper. See \cite{DBLP:dblp_journals/ml/AndrieuFDJ03,Zoutendijk76,Minoux86}
for more details.

 A Mathematical Optimization (MO) problem is usually  formulated as:
\begin{equation} \label{eq:background:1}
 \begin{aligned}
 & {\text{minimize}}
 & & f(x) \\
 & \text{subject to}
 & & x\in S
 \end{aligned}
 \end{equation}
 where $f$ is called the {objective function}, and $S$ the {search
   space}.
 In general, mathematical optimization problems can be divided into
 two categories. One focuses on how functions are shaped at local
 regions and where a local minimum can be found near a given
 input. This \emph{local optimization} is classic, usually involving
 standard techniques such as Newton's or the steepest descent methods.
 Local optimization not only provides the minimum value of a function
 within a neighborhood of the given input points, but also aids
 \emph{global optimization}, which determines the function minimum
 over the entire search space.

\Paragraph{Local Optimization.}

 Let $f$ be a function defined over a Euclidean space with distance
 function $d$.  We call $x^*$ a \emph{local minimum point} if there exists
 a \emph{neighborhood} of $x^*$, namely $\{x \mid d(x,x^*) < \delta\}$
 for some $\delta>0$, so that all $x$ in the neighborhood satisfy
 $f(x)\geq f(x^*)$.  The value $f(x^*)$ is called a \emph{local
   minimum} of $f$.

 Local optimization problems can usually be efficiently solved if the
 objective function is smooth (such as continuous or differentiable to
 some degree)~\cite{citeulike:2621649}.
 Fig.~\ref{fig:background:1}(a) shows a common local optimization
 method with the objective function $\lambda x. x\leq 1~?~0:(x-1)^2$.
 It uses tangents of the curve to quickly converge to a minimum
 point. The smoothness of the curve makes it possible to deduce the function's
 behavior in the neighborhood of a particular point $x$ by using
 information at $x$ only.

\Paragraph{Global Optimization and MCMC.}
If $f(x^*)\leq f(x)$ for all $x$ in the search space, we call $x^*$ a
\emph{global minimum point} (or minimum point for short), and 
 $f(x^*)$ the \emph{global minimum} (or minimum for short) of the function $f$. In this presentation, if we say ``$x^*$ minimizes the function $f$'',  we mean $x^*$ is a global minimum point of $f$. 



 \begin{figure}[t]
\begin{subfigure}[b]{0.45\linewidth}
   \centering
  \includegraphics[width=\linewidth]{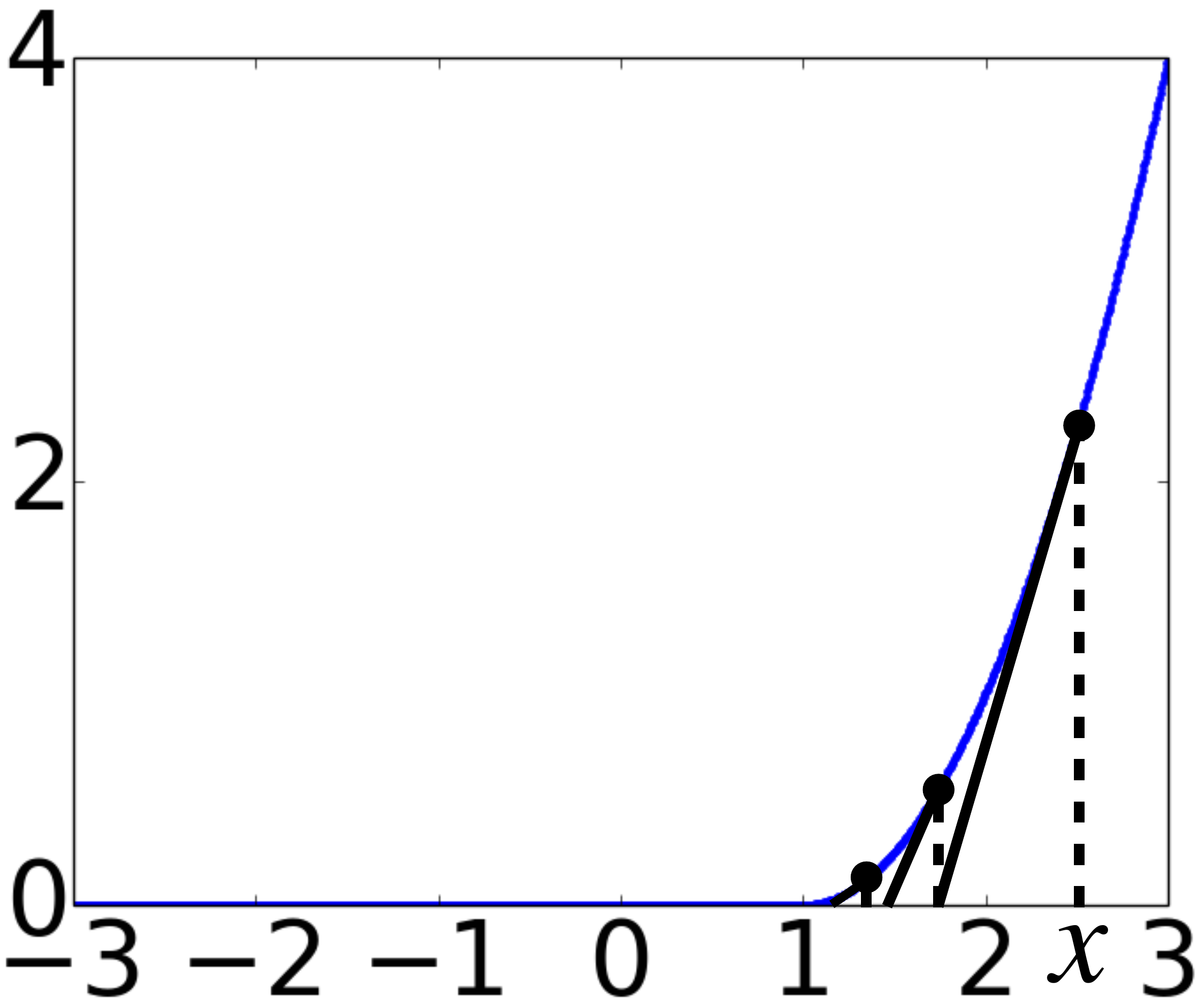}
  \caption{}
\end{subfigure}%
\hfill
\begin{subfigure}[b]{0.46\linewidth}
   \centering
  \includegraphics[width=\linewidth]{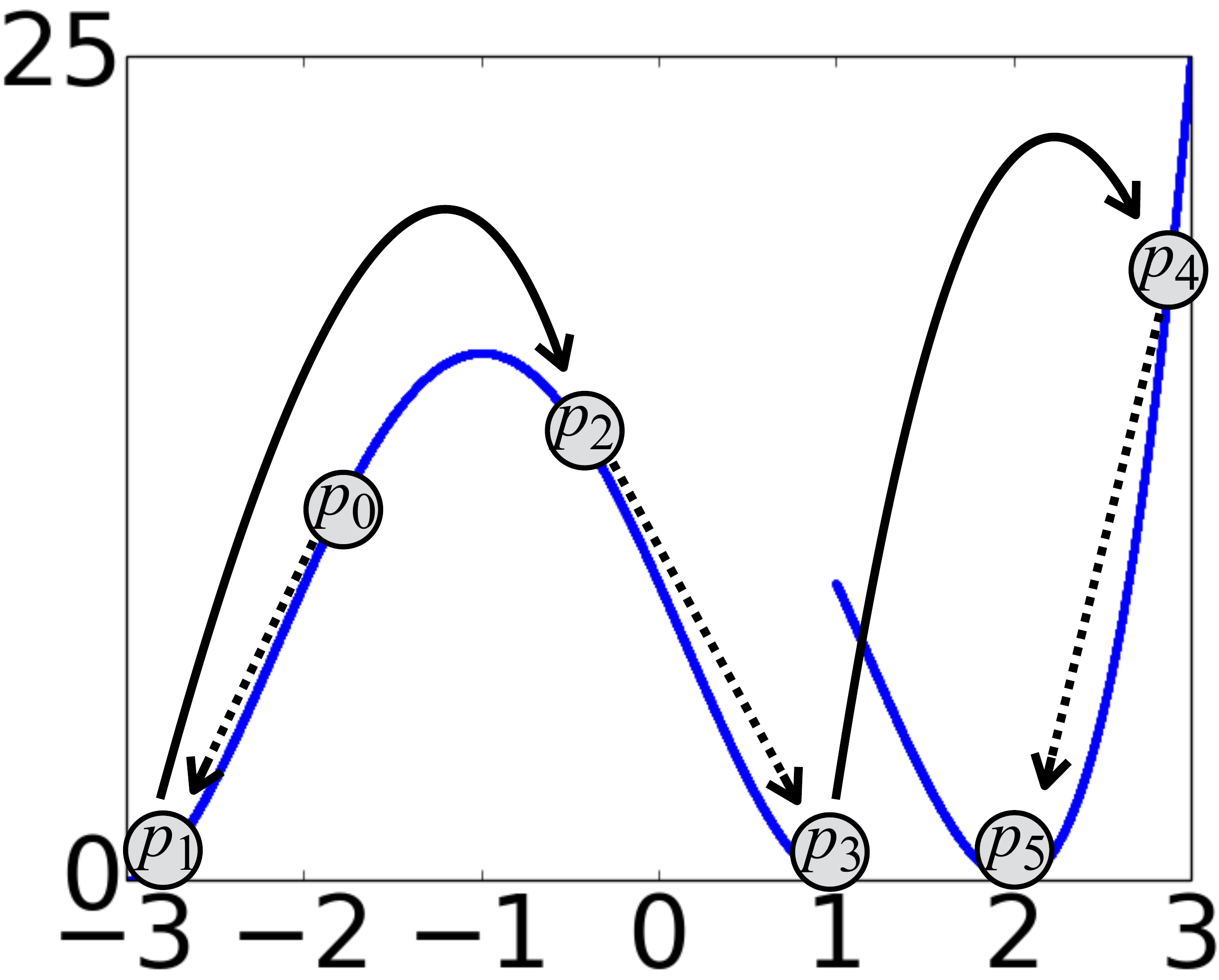}
  \caption{}
\end{subfigure}
\caption{ (a) Local optimization with the curve of
  $\lambda x. x\leq 1~?~0:(x-1)^2$.  The illustrated technique uses
  tangents of the curve to converge quickly  to a minimum point; (b)
  Global optimization with the curve of
  $\lambda x.x \leq 1~?~((x+1)^2-4)^2:(x^2-4)^2$.  The MCMC method starts
  from $p_0$, converges to local minimum $p_1$, performs a Monte-Carlo
  move to $p_2$ and converges to $p_3$. Then it moves to $p_4$ and
  finally converges to $p_5$.  }
\label{fig:background:1}
 \end{figure}


 We use the Monte Carlo Markov Chain (MCMC)
 sampling to solve global optimization
 problems.  
A fundamental fact regarding MCMC is
 that it follows the target distribution asymptotically.  For
 simplicity, we give the results~\cite{DBLP:dblp_journals/ml/AndrieuFDJ03}  with the discrete-valued probability.
\begin{lem}
  Let $x$ be a random variable, $A$ be an  enumerable set of the
  possible values of $x$. Let $f$ be a target probability  distribution for $x$, \ie,  the probability of $x$ taking value  $a$ is  $f(a)$.  Then, for an MCMC sampling sequence
  $x_1,\ldots,x_n\ldots$ and a probability density function
  $P(x_n=a)$ for each $x_n$, we have $P(x_n =a) \rightarrow f(a)$.
\label{lem:pdf}
\end{lem}
 
For example, consider the target distribution of coin tossing with
$0.5$ probability for having the head. An MCMC sampling is a sequence
of random variables $x_1$,\ldots, $x_n,\ldots$, such that the probability of
$x_n$ being ``head'', denoted by $P_n$, converges to $0.5$.


MCMC provides multiple advantages in practice.  Because such sampling
can simulate an arbitrary distribution (Lem.~\ref{lem:pdf}), MCMC
backends can sample for a target distribution in the form of
$\lambda x. \exp^{-f(x)}$ where $f$ is the function to minimize, which
allows its sampling process to attain the minimum points more
frequently than the other points.  Also, MCMC has many
sophisticated techniques that integrate with classic local search
techniques, such as the Basinhopping
algorithm~\cite{leitner1997global} mentioned above.  Some variants of
MCMC can even handle high dimensional
problems~\cite{robbins1951stochastic}, or non-smooth objective
functions~\cite{eckstein1992douglas}.  Fig.~\ref{fig:background:1}(b)
illustrates a typical MCMC cycle.  Steps $p_0\rightarrow p_1$,
$p_2\rightarrow p_3$, and $p_4 \rightarrow p_5$ are the local
optimization; Steps $p_1\rightarrow p_2$ and $p_3\rightarrow p_4$ aim
to prevent the MCMC sampling from getting trapped in the local minima.

\section{Branch Coverage Based Testing}
 \label{sect:overview}

This section shows a detailed ME procedure in solving branch coverage
based testing for numerical code. 


\begin{figure}[t]
\centering
\includegraphics[width=1.0\linewidth]{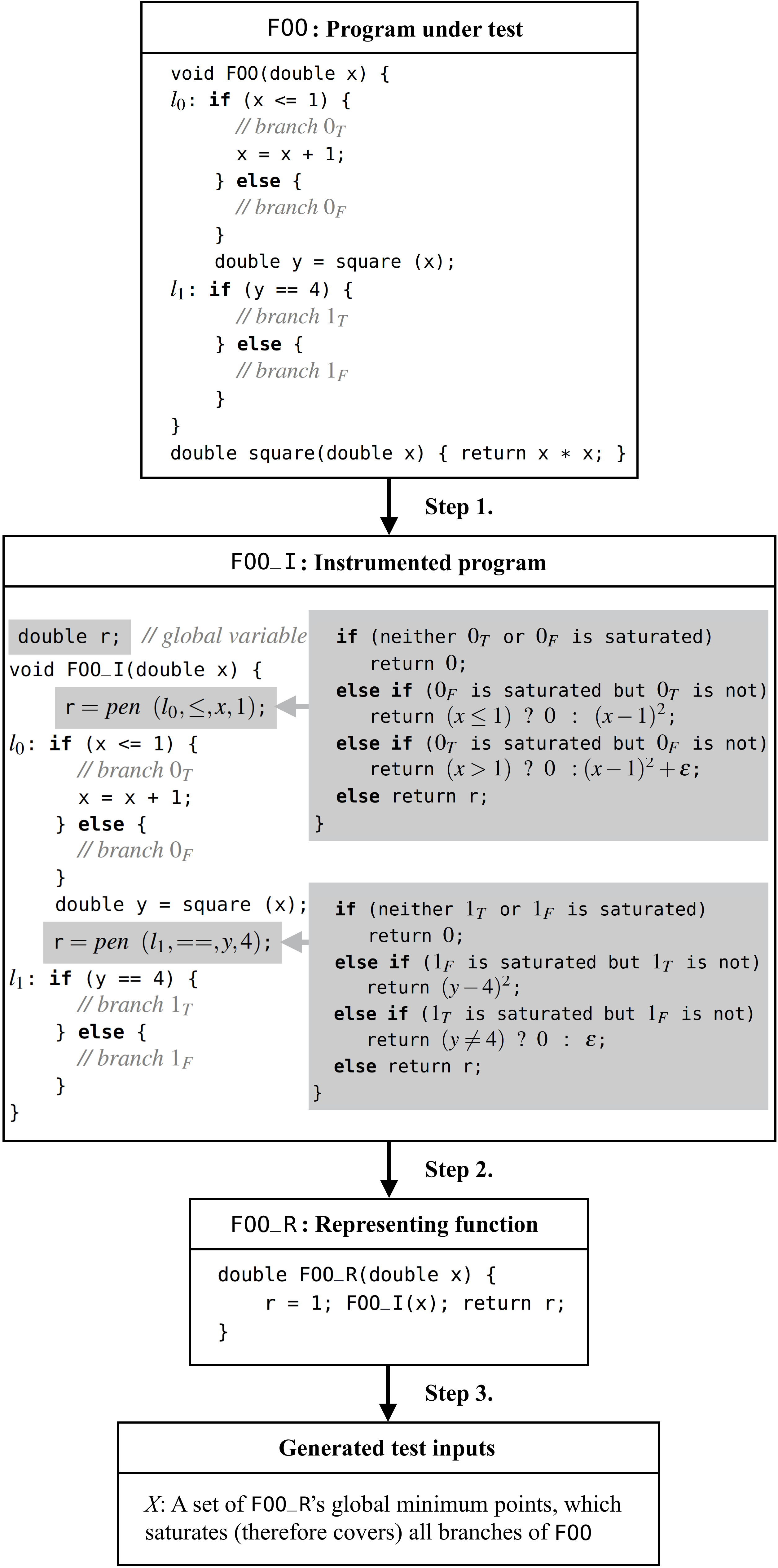}
\caption{Branch coverage based testing via $\ME$. The goal is to
  saturate (therefore cover) all branches of \FOO, \ie,
  $\{0_T,0_F,1_T,1_F\}$. }
\label{fig:algo:injecting_pen}
\end{figure}

\subsection{Problem Statement}
\label{sect:overview:pb}

\begin{definition}
  
  Let \FOO be the program under test with $N$ conditional statements,
  labeled by $l_0$,\ldots,$l_{N-1}$.  Each $l_i$ has a {\tt true}
  branch $i_T$ and a {\tt false} branch $i_F$.  The problem of \emph{branch
  coverage based testing} aims to find a set of inputs
  $X\subseteq \dom{\FOO}$ that \emph{covers} all $2*N$ branches of
  \FOO.  Here, we say a branch is ``covered'' by $X$ if it is passed
  through by the path of executing \FOO with an $x\in X$.  We
  scope the problem with three assumptions:
\begin{itemize}
\item[(a)]  The inputs of  \FOO are floating-point numbers.
\item[(b)]  Each Boolean condition in \FOO  is an arithmetic comparison in the form of $a~op~b$, where $\myop\in\{==,\leq,<,\neq,\geq,>\}$, and $a$ and $b$ are floating-point variables or constants.
\item[(c)]   Each branch of  \FOO is feasible, \ie, it is covered by $\dom{\FOO}$.
\end{itemize}
Assumptions (a) and (b) are set for modeling numerical
code. Assumption (c) is set to simplify our presentation. Our
implementation will partially relax these assumptions (details in
Appendix~\ref{sect:implem}).

\label{def:overview:branchcoverage}
\end{definition}






We introduce the concept of a \emph{saturated branch} and use it to reformulate Def.~\ref{def:overview:branchcoverage}.
\begin{definition}
  Let $X$ be a set of inputs generated during the process of testing.
  We say a branch is \emph{saturated} by $X$ if the branch itself and
  all its descendant branches, if any, are covered by $X$. Here,  a branch
  $b'$ is called a descendant branch of $b$ if there exists a
  segment of control flow path from $b$ to $b'$.  Given  $X\subseteq\dom{\FOO}$, we write
\begin{align}
\Explored(X)
\end{align}
for the set of branches saturated by $X$. 

\label{def:overview:saturate}
\end{definition}

\noindent  \begin{minipage}{0.69\linewidth}
\paragraph{}
To illustrate Def.~\ref{def:overview:saturate}, suppose that an input
set $X$ covers $\{0_T,0_F,1_F\}$ in the control-flow graph on the
right, then $\Explored(X)= \{0_F,1_F\}$; $1_T$ is not saturated
because it is not covered, and $0_T$ is not saturated because its
descendant branch $1_T$ is not covered.

  \end{minipage}%
  \begin{minipage}{0.3\linewidth}
\flushright   \includegraphics[width=0.95\linewidth]{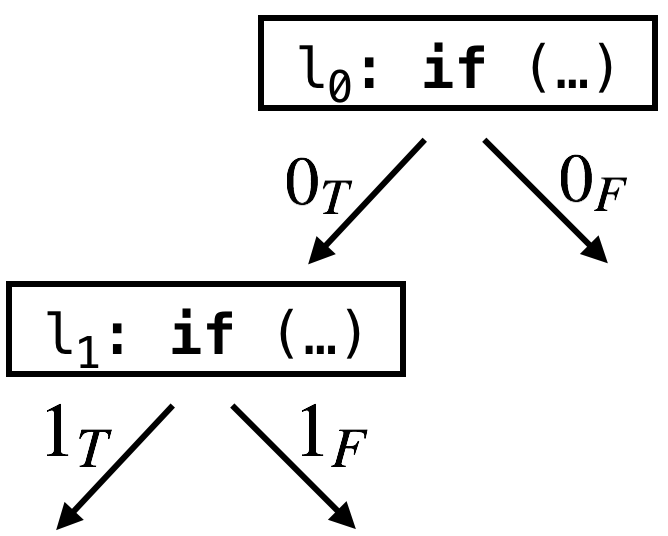}
\end{minipage}


\paragraph{}
Observe that if \emph{all} \FOO's branches are covered by an input set
$X$, they are also saturated by $X$, and vice versa. This observation
allows us to reformulate the branch coverage based testing problem following the lemma below.

\begin{lem} Let \FOO be a program under test with the assumptions Def.~\ref{def:overview:branchcoverage}(a-c) satisfied.  \emph{The goal of branch coverage based testing}  can be  stated as to find a small set of inputs
  $X\subseteq \dom{\FOO}$ that saturates all \FOO's branches.
\label{def:overview:pb}
\end{lem}

\begin{remark}
In Lem.~\ref{def:overview:pb}, we expect the generated input set $X$ to be ``small'', because otherwise we can use $X=\dom{\FOO}$ which already saturates all branches under the assumption Def.~\ref{def:overview:branchcoverage}(c).
\end{remark}

\subsection{An Example}
\label{sect:overview:example}

We use a simple program \FOO in Fig.~\ref{fig:algo:injecting_pen} to
illustrate our approach.  The program has two conditional statements $l_0$
and $l_1$, and their {\tt true} and {\tt false} branches are denoted
by $0_T$, $0_F$ and $1_T$, $1_F$ respectively. The objective is to
find an input set that saturates all branches.  Our approach proceeds
in three steps:

\Paragraph{Step 1.}
We inject a global variable $\myr$ in \FOO, and, immediately before
each control point $l_i$, we inject an assignment
(Fig.~\ref{fig:algo:injecting_pen})
\begin{align}
\myr = \pen,
\end{align}
where $\pen$ invokes a code segment with parameters specific to $l_i$.
The idea of $\pen$ is to capture the \emph{distance of the program
  input from saturating a branch that has not yet been saturated}.  As
illustrated in Fig.~\ref{fig:algo:injecting_pen}, this distance
returns different values depending on whether the branches at $l_i$ are
saturated or not.

We denote the instrumented program by \FOOI. The key in Step 1 is to
design $\pen$ to meet certain conditions that allow us to approach the
problem defined in Lem.~\ref{def:overview:pb} as a mathematical optimization
problem. We will specify the conditions  in the next step.





\Paragraph{Step 2.} This step constructs the representing function that we have mentioned in Sect.~\ref{sect:intro}.  The representing function is the driver program \FOOR shown in 
Fig.~\ref{fig:algo:injecting_pen}. It initializes $\myr$ to $1$, invokes \FOOI and then returns $\myr$ as the output of \FOOR. That is, $\FOOR(x)$ for a given  input $x$  calculates the value of $\myr$ at the end of executing $\FOOI(x)$.

Our approach requires two conditions on \FOOR:
\begin{description}[align=left]
\item [C1.] $\FOOR(x)\geq 0$ for all $x$, and 
\item [C2.]   $\FOOR(x)=0$ if and only if $x$ saturates a new branch.  In other
words, a branch that has not been saturated by the generated input set
$X$ becomes saturated with $X\cup \{x\}$, \ie,
$\Explored(X)\neq \Explored(X\cup\{x\})$.
\end{description}
Imagine that we have designed $\pen$ so that \FOOR meets both C1 and
C2. Ideally, we can then saturate all branches of \FOO by repeatedly
minimizing \FOOR as shown in the step below.




\Paragraph{Step 3.}
In this step, we use MCMC to calculate the minimum points of
\FOOR. Other mathematical optimization techniques, \eg, genetic
programming~\cite{koza1992genetic}, may also be applicable, which we
leave for future investigation.  

We start with an input set $X=\emptyset$ and
$\Explored(X)=\emptyset$. We minimize \FOOR and obtain a minimum point
$x^*$ which necessarily saturates a new branch by condition C2. Then
we have $X=\{x^*\}$ and we minimize \FOOR again which gives another
input $x^{**}$ and $\{x^*,x^{**}\}$ saturates a branch that is not saturated by $\{x^*\}$.  We continue this process until all
branches are saturated. When the algorithm terminates, $\FOOR(x)$ must
be strictly positive for any input $x$, due to C1 and C2.




\paragraph{}
Tab.~\ref{tab:algo:example_illustration} illustrates a scenario of how
our approach saturates all branches of \FOO. Each ``\textbf{\#n}''
below corresponds to one line in the table.  We use $\pen_0$ and $\pen_1$
to denote $\pen$ injected at $l_0$ and $l_1$ respectively.
(\textbf{\#1}) Initially, $\Explored=\emptyset$. Any input saturates a
new branch.  Both $\pen_0$ and $\pen_1$ set $\myr=0$, and
$\FOOR=\lambda x.0$ (Fig.~\ref{fig:algo:injecting_pen}). Suppose
$x^*=0.7$ is found as the minimum point.
(\textbf{\#2})  The branch $1\myF$ is now saturated and $1\myT$ is not. Thus, $\pen_1$ sets $\myr = (y-4)^2$.  Minimizing $\FOOR$ 
gives $x^*=-3.0$, $1.0$, or $2.0$. We have illustrated this MCMC procedure in Fig.~\ref{fig:background:1}(b).  Suppose $x^*=1.0$ is found.
(\textbf{\#3}) Both $1\myT$ and $1\myF$, as well as $0_T$, are
saturated by the generated inputs $\{0.7,1.0\}$.  Thus, $\pen_1$
returns the previous $\myr$.  Then, $\FOOR$ amounts to $\pen_0$,
returning $0$ if $x>1$, or $(x-1)^2 + \epsilon$ otherwise, where $\epsilon$ is a  small positive constant.  Suppose
$x^* = 1.1$ is found as the minimum point.
(\textbf{\#4}) All branches have been saturated. In this case, both
$\pen_0$ and $\pen_1$ return the previous $\myr$. Then, $\FOOR$
becomes $\lambda x.1$, understanding that $\FOOR$ initializes $\myr$
as $1$.  The minimum point, \eg,  $x^*=-5.2$, necessarily
satisfies $\FOOR(x^*)>0$, which terminates the algorithm.

\begin{table}\footnotesize 
\setlength{\tabcolsep}{1.5pt}%
  \centering
  \caption{
    A scenario of how our approach saturates all branches of \FOO by repeatedly minimizing
    \FOOR.  Column ``$\Explored$'':  Branches that have been saturated. Column ``$\FOOR$'': The representing function and its plot. Column ``$x^*$'': The point where $\FOOR$ attains the minimum.  Column ``$\myX$'': Generated test inputs.}\label{tab:algo:example_illustration}. 
\scriptsize
  \begin{tabular}{l ccccc}
\toprule
 \#~~~&  $\Explored$  &$\FOOR$ & &  $x^*$  & $\myX$ \\ \midrule 
1& $\emptyset$ & $\lambda x. 0$ &\raisebox{-.5\height}{\includegraphics[width=.27\linewidth]{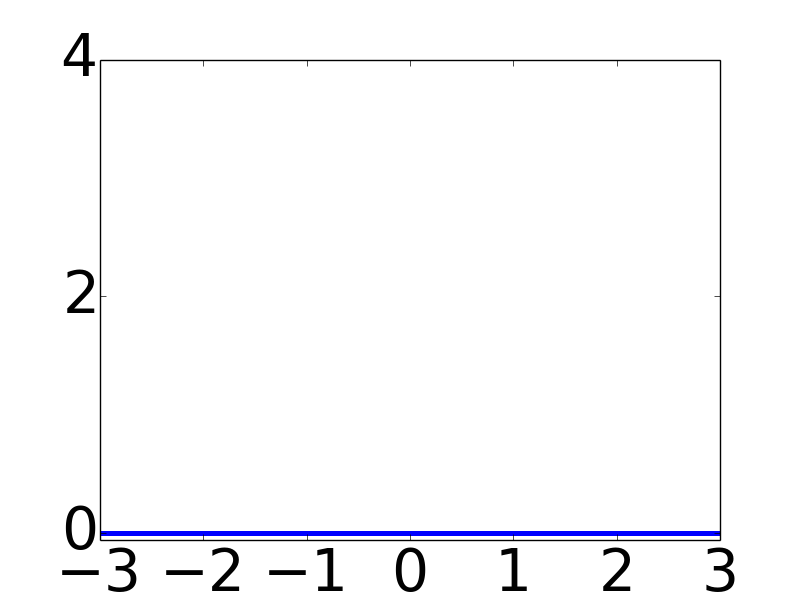}} &$0.7$ & $\{0.7\}$ \\  \arrayrulecolor{light-gray}\hline
2& $\{1\myF\}$ & $\lambda x.\begin{cases}((x+1)^2-4)^2& x\leq 1 \\ (x^2-4)^2 & \text{else} \end{cases}$ &\raisebox{-.5\height}{\includegraphics[width=.27\linewidth]{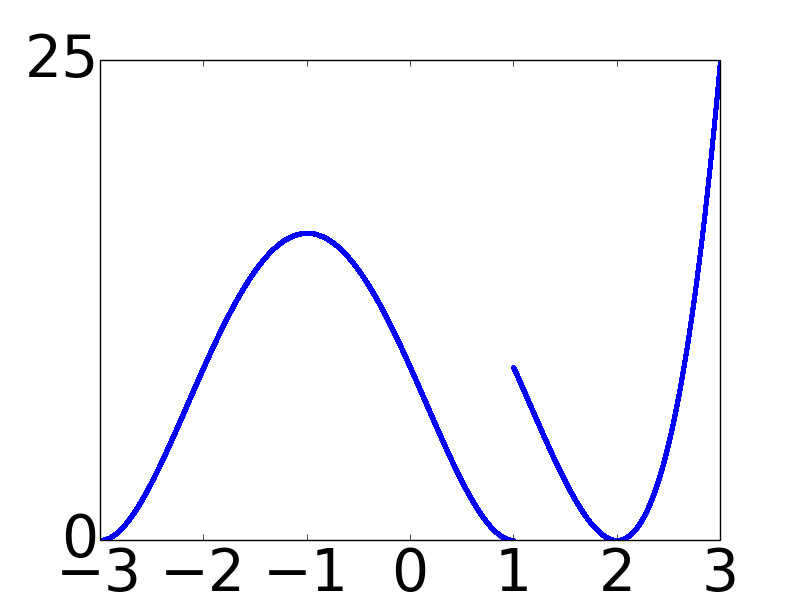}} &$1.0$ & $\{0.7,1.0\}$ \\  \hline
3& \parbox{3em}{$\{0\myT,1\myT,\\ 1\myF\}$} & $\lambda x. \begin{cases}0 & x>1 \\ (x-1)^2+\epsilon &\text{else} \end{cases}$ &\raisebox{-.5\height}{\includegraphics[width=.27\linewidth]{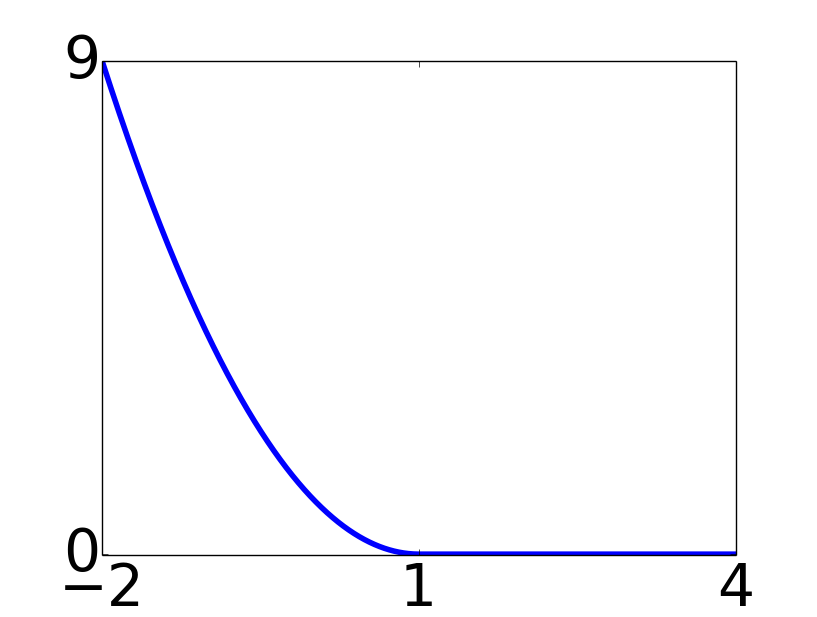}} &$1.1$ & \parbox{3em}{$\{0.7,1.0,$ \\ $1.1\}$} \\  \hline
4&\parbox{3em}{$\{0\myT,1\myT,\\ 0\myF, 1\myF\}$} & $\lambda x. 1$ &\raisebox{-.5\height}{\includegraphics[width=.27\linewidth]{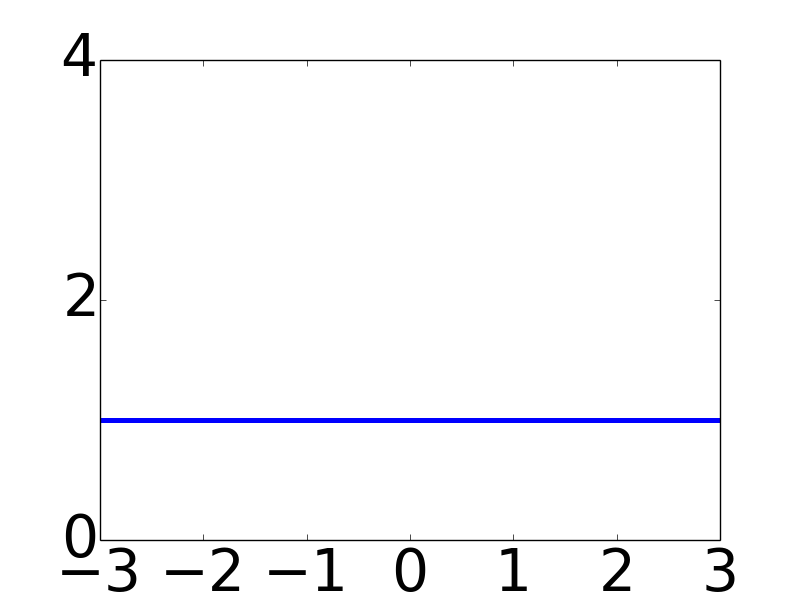}} &$-5.2$ & \parbox{3em}{$\{0.7,1.0,\\1.1,-5.2\}$} \\  
\arrayrulecolor{black}
\bottomrule
  \end{tabular}

\end{table}

\begin{remark}

Given an input $x$, the value of $\FOOR(x)$ may change during the
minimization process.  In fact, \FOOR is constructed with
injected $\pen$ which returns different values at $l_i$ depending on whether the
branches $i_T$ and $i_F$ have been saturated. Thus, the minimization step in
our algorithm differs from existing  mathematical optimization techniques where the objective function is
fixed~\cite{Zoutendijk76}.

\end{remark}



\subsection{Algorithm}
\label{sect:theory:algo}
We provide details corresponding to the three steps in
Sect.~\ref{sect:overview:example}. The algorithm is summarized in
Algo.~\ref{theory:algo:coverme}.

\Paragraph{Algorithm for Step 1.} 
The outcome of this step is the instrumented program \FOOI.  As
explained in Sect.~\ref{sect:overview:example}, the essence is to
inject the variable $\myr$ and the assignment $\myr=\pen$ before each conditional
statement (Algo.~\ref{theory:algo:coverme}, Lines 1-4).



To define $\pen$, we first introduce a set of helper functions that
are sometimes known as \emph{branch distance}. There are many
different forms of branch distance in the
literature~\cite{Korel:1990:AST:101747.101755,McMinn:2004:SST:1077276.1077279}.
We define ours with respect to an arithmetic condition $a~op~b$.

\begin{definition}\label{def:overview:bd}
Let $a,b\in\Real$, $\myop\in\{==,\leq,<,\neq,\geq,>\}$, $\epsilon\in\Real_{> 0}$.
 We define branch distance $d_{\epsilon}(\myop, a,b)$ as follows:
\begin{align}
d_{\epsilon}(==, a,b)&\mydef (a-b)^2 \\
d_{\epsilon}(\leq, a,b)&\mydef (a\leq b)~?~0:(a-b)^2\\
d_{\epsilon}(<,a,b)&\mydef (a <b)~?~0:(a-b)^2+\epsilon\\
d_{\epsilon}(\neq, a,b)&\mydef (a\neq b)~?~0:\epsilon 
\end{align}
and $d_{\epsilon}(\geq,a,b)\mydef d_{\epsilon}(\leq,b,a)$,
$d_{\epsilon}(>,a,b)\mydef d_{\epsilon}(<,b,a)$. 
Usually, the parameter $\epsilon$ is a small
constant, so  we drop the explicit reference to $\epsilon$ when
using the branch distance.
\end{definition}

The intention of $d(\myop,a,b)$ is to {quantify how far $a$ and $b$
  are from attaining $a~\myop~b$}. For example, $d(==,a,b)$ is
strictly positive when $a\neq b$, becomes smaller when $a$ and $b$
go closer, and vanishes when $a==b$. The following property holds:
\begin{align}
d(\myop,a,b)\geq 0 \text{~and~} d(\myop,a,b)=0\Leftrightarrow a~\myop~b.
\label{eq:overview:brDist}
\end{align}

 As an analogue, we set $\pen$
to {quantify how far an input is from saturating a new branch}.
We define $\pen$ following Algo.~\ref{theory:algo:coverme}, Lines~14-23.
\begin{definition}
For  branch coverage based testing, the function $\pen$ has four parameters, namely, the label of the
conditional statement $l_i$, and $\myop$, $a$ and $b$ from the arithmetic
condition $a~\myop~b$.
\begin{itemize}
\item [(a)]  If neither of the two branches at $l_i$ is
saturated, we let $\pen$ return $0$ because any input saturates a
new branch (Lines~16-17). 
\item [(b)] If one branch at $l_i$ is saturated but the other is
not, we set $\myr$ to be the distance to the unsaturated branch (Lines~18-21).
\item [(c)] If both branches at $l_i$ have already been saturated,
$\pen$ returns the previous value of the global variable $\myr$ (Lines~22-23).
\end{itemize}

  For example, the two instances of $\pen$ at $l_0$ and $l_1$ are
invoked as $pen(l_i, \leq,x,1)$ and $\pen(l_1,==,y,4)$ respectively in
Fig.~\ref{fig:algo:injecting_pen}.

\label{def:overview:pen}
\end{definition}


\SetKwProg{myproc}{Procedure}{}{} \SetKwInput{Gb}{Global}
\begin{algorithm} [t]\footnotesize 
  \DontPrintSemicolon
  \KwIn{\begin{tabularx}{.8\linewidth}[t]{>{$}l<{$} X}
      \FOO & Program under test \\
      \nStart & Number of starting points\\
      \LM & Local optimization  used in \mcmc\\
      \niter & Number of iterations for \mcmc\\
    \end{tabularx}} 

\KwOut{\begin{tabularx}{.8\linewidth}[t]{>{$}l<{$} X}
      X & Generated input set \\
    \end{tabularx}}

\BlankLine
  \tcc{Step 1}
Inject   global variable $\myr$  in $\FOO$\; 
\For{ conditional
    statement $l_i$ in {\upshape $\FOO$} }{ 
    Let the  Boolean condition at  $l_i$ be 
    $\mylhs~\myop~\myrhs$ where $\myop\in\{\leq,<,=,>,\geq,\neq\}$\;
    Insert assignment $\myr=\pen(l_i,\myop,\mylhs, \myrhs)$ before $l_i$ }
  \tcc{Step 2}
  Let $\FOOI$ be the newly instrumented program, and $\FOOR$ be:
  \lstinline!  double FOO_R(double x) {r = 1; FOO_I(x); return r;}!\;

\tcc{Step 3}
Let $\Explored=\emptyset$\;
Let $X =\emptyset$ \; 
\For{$k=1$ to $\nStart$}{
Randomly take a  starting point $x$\;
Let $x^* = \text{\mcmc}(\FOOR,x)$\;
  \lIf{\upshape $\FOOR(x^*) = 0$}{ $ X = X \cup \{x^*\}$  }
Update $\Explored$ \;
}

\Return $X$\;

  \BlankLine

\SetKwProg{Fn}{Function}{}{}

\Fn{$\pen(l_i,\myop, \mylhs,\myrhs)$}{
Let $i\myT$ and $i\myF$ be the {\tt true} and the {\tt false} branches at $l_i$\;
\If{$i\myT\not\in\Explored$ and $i\myF\not\in\Explored$ }
{\Return $0$}
\uElseIf{$i\myT\not\in\Explored$ and $i\myF\in\Explored$ }
{\Return $\dist(\myop,\mylhs,\myrhs)$ 
 \tcc{$d$: Branch distance}
}
\uElseIf{$i\myT\in\Explored$ and $i\myF\not\in\Explored$ }
{\Return $\dist(\overline{\myop}, \mylhs,\myrhs)$  \tcc{$\overline{\myop}$: the opposite of $\myop$}}
\uElse(\tcc*[h]{$i\myT\in\Explored$ and $i\myF\in\Explored$ }){\Return $\myr$ }
}

  \BlankLine
 \Fn{{\upshape \mcmc}($\energy$, $x$)}
 {

     $\xloc=\lmin(\energy, x)$\; \tcc{Local minimization}  

     \For{$k = 1$ \KwTo $\niter$}{

       Let $\delta$ be a random perturbation generation from a predefined distribution\;
       Let  $\xpro=\lmin(\energy, \xloc + \delta)$\;

       \lIf{$\energy(\xpro)<\energy(\xloc)$}{$\mathit{accept} = \mathit{true}$}

       \Else{

         Let $m$ be a random number  generated from the uniform distribution on $[0,1]$\;

         Let ${\mathit{accept}}$ be the Boolean $m<\exp(\energy(\xloc) - \energy(\xpro)) $ 
        }

       \lIf{$\mathit{accept}$}{ $\xloc = \xpro$}
     }    

     \Return $\xloc$\;
 }

  \caption{Branch coverage based testing via ME.}
  \label{theory:algo:coverme} 
\end{algorithm}

\Paragraph{Algorithm for Step 2.}
This step constructs the representing function \FOOR
(Algo.~\ref{theory:algo:coverme}, Line~5). Its input domain is the
same as that of \FOOI and \FOO, and its output domain is {\tt double},
so to simulate a real-valued mathematical function which can then be
processed by the mathematical optimization  backend.

\FOOR initializes $\myr$ to 1. This is essential for the correctness of the
algorithm because we expect \FOOR returns a non-negative value when all branches
are saturated (Sect.~\ref{sect:overview:example}, Step 2).
\FOOR then calls $\FOOI(x)$ and records the value of $\myr$ at the end
of executing $\FOOI(x)$. This $\myr$ is the returned value of \FOOR.

As mentioned in Sect.~\ref{sect:overview:example}, it is important to
ensure that \FOOR meets conditions C1 and C2.  The condition {C1}
holds true since \FOOR returns the value of the instrumented $\myr$, which
is never assigned a negative quantity. The lemma below states \FOOR
also satisfies C2.

\begin{lem}
Let \FOOR be the program constructed in Algo.~\ref{theory:algo:coverme},  and $S$ the branches that have been saturated. Then, for any input $x\in\dom{\FOO}$, $\FOOR(x)=0$ $\Leftrightarrow $  $x$ saturates a branch that does not belong to $S$. 
\label{lem:overview:c2}
\end{lem}
\begin{proof}
  We first prove the $\Rightarrow$ direction.  Take an arbitrary $x$
  such that $\FOOR(x)=0$. Let $\tau=[l_0,\ldots l_n]$ be the path in
  $\FOO$ passed through by executing $\FOO(x)$. We know, from Lines
  2-4 of the algorithm, that each $l_i$ is preceded by an invocation
  of $\pen$ in $\FOOR$. We write $\pen_i$ for the one injected before
  $l_i$ and divide $\{\pen_i\mid i\in[1,n]\}$ into three groups. For
  the given input $x$, we let {\it P1}, {\it P2} and {\it P3} denote the groups of $\pen_i$ that
  are defined in Def.~\ref{def:overview:pen}(a),
  (b) and (c),
  respectively.  Then, we can always have a prefix path of
  $\tau=[l_0,\ldots l_m]$, with $0\leq m\leq n$ such that each
  $\pen_i$ for $i\in[m+1,n]$ belongs to {\it P3}, and each $\pen_i$ for
  $i\in [0,m]$ belongs to either {\it P1} or {\it P2}.  Here, we can guarantee the
  existence of such an $m$ because, otherwise, all $\pen_i$ belong in
  {\it P3}, and $\FOOR$ becomes $\lambda x.1$. The latter contradicts the
  assumption that $\FOOR(x)=0$.
  Because each $\pen_i$ for $i>m$ does nothing but performs
  $\myr=\myr$, we know that $\FOOR(x)$ equals to the exact value of
  $\myr$ that $\pen_m$ assigns. Now consider two disjunctive cases on
  $\pen_m$. If $\pen_m$ is in {\it P1}, we immediately conclude that
  $x$ saturates a  new branch. Otherwise, if $\pen_m$ is in {\it P2}, we
  obtains the same from Eq.~\eqref{eq:overview:brDist}.
Thus,  we have established the $\Rightarrow$
  direction of the lemma.

  To prove the $\Leftarrow$ direction, we use the same
  notation as above, and let $x$ be the input that saturates a new branch,
  and $[l_0,\ldots, l_n]$ be the exercised path. Assume that $l_m$
  where $0\leq m\leq n$ corresponds to the newly saturated branch. We
  know from the algorithm  that (1) $\pen_m$ updates
  $\myr$ to $0$, and (2) each $\pen_i$ such that $i>m$ maintains the
  value of $\myr$ because their descendant branches have been saturated.
We have thus proven the $\Leftarrow$ direction of the lemma.
\end{proof}



\Paragraph{Algorithm for Step 3.}
The main loop (Algo.~\ref{theory:algo:coverme}, Lines 8-12) relies on
an existing MCMC engine. It takes an objective function and a starting
point and outputs $x^*$ that it regards as a minimum point.  Each
iteration of the loop launches MCMC from a randomly selected starting
point (Line 9). From each starting point, MCMC computes the minimum
point $x^*$ (Line 10). If $\FOOR(x^*)=0$, $x^*$ is added to the set of
the generated inputs $X$ (Line 11).  Lem.~\ref{lem:overview:c2}
ensures that $x^*$ saturates a new branch in the case of
$\FOOR(x^*)=0$. Therefore, in theory, we only need to set
$\nStart=2*N$ where $N$ denotes the number of conditional statements,
so to saturate all $2*N$ branches.  In practice, however, we set
$\nStart>2*N$ because MCMC cannot guarantee that its output is a true
global minimum point.


The MCMC procedure (Algo.~\ref{theory:algo:coverme}, Lines 24-34) is
also known as the Basinhopping
  algorithm~\cite{leitner1997global}. It is an MCMC sampling over the
space of the local minimum points~\cite{Li01101987}.  The random
starting point $x$ is first updated to a local minimum point $\xloc$
(Line 25). Each iteration (Lines 26-33) is composed of the two phases
that are classic in the {Metropolis-Hastings} algorithm family of
MCMC~\cite{citeulike:831786}. In the first phase (Lines 27-28), the
algorithm \emph{proposes} a sample $\xpro$ from the current sample
$x$. The sample $\xpro$ is obtained with a perturbation $\delta$
followed by a local minimization, \ie, $\xpro=\LM(f,\xloc+\delta)$
(Line 28), where \LM denotes a local minimization in Basinhopping, and
$f$ is the objective function.  The second phase (Lines 29-33) decides
whether the proposed $\xpro$ should be accepted as the next sampling
point. If $\energy(\xpro)<\energy(\xloc)$, the proposed $\xpro$ will
be sampled; otherwise, $\xpro$ may still be sampled, but only with the
probability of $\exp((\energy(\xloc)-\energy(\xpro))/T)$, in which $T$
(called the annealing temperature~\cite{Kirkpatrick83optimizationby})
is set to $1$ in Algo.~\ref{theory:algo:coverme} for simplicity.

\section{Mathematical Execution}
\label{sect:theory}
\begin{figure*}
\includegraphics[width=1\linewidth]{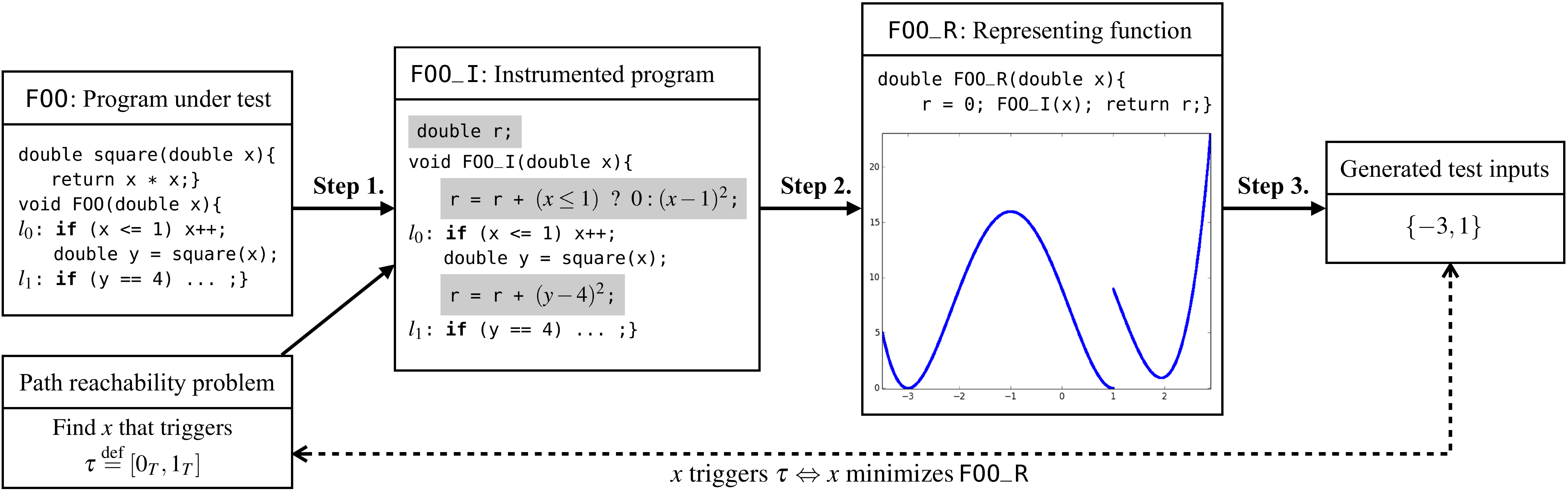}
\caption{Path reachability testing via $\ME$. The goal of this example is to find a
  test input that triggers the path $[0\myT,1\myT]$ of the program
  $\FOO$.}\label{fig:overview:ME_illustration}
\end{figure*}

\emph{Predicate as function} is a common
concept in mathematics.  As G\"odel stated in his 1931
work~\cite{Davis:2004:UBP:1098684,godel1931formal},
\begin{quote}
  There shall correspond to each relation $R$
  ($R\subseteq \Integer^n$) a representing function
  $\phi(x_1,\ldots, x_n)=0$ if $R(x_1,\ldots,x_n)$ and
  $\phi(x_1,\ldots,x_n)=1$ if $\neg R(x_1\ldots x_n)$.
\end{quote}
Traditional representing function in its Boolean nature is a
predicate/set indicator of two states, essentially being true or
false.  For example, $\mathnormal{even} (N)$ that decides whether an
integer $N$ is even can be represented by the function
$\lambda x. (x\ \text{mod}\ 2 ==0)$.

In this section, we present Mathematical Execution ($\ME$) by
extending the Boolean-valued representing function to a real-valued
calculus 
to  address  a spectrum of automated testing problems for numerical code, which we unify in the category of the \emph{search problem}.

\subsection{Search Problem}

\begin{definition}
The \emph{search problem} with regard to a set $X$ aims to 
\begin{itemize}
\item[(a)] find an $x\in\myX$ if $\myX\neq\emptyset$, and
\item[(b)]  report ``not found'' if  $X=\emptyset$.
\end{itemize}
Usually, we have a search space $\myU$, and $\myX$ is specified implicitly as a subset of  $\myU$. We denote the search problem by $(\myX,\myU)$.  In this paper, we deal with numerical code, and thus,  we assume that  $\myX$ is a subset of $\Real^N$. We also assume that $\myX$  is decidable so that we  can check whether an  $x\in\myU$ is an element of $\myX$.
\label{def:theory:search}
\end{definition}

\begin{example}
  A search problem can be any computational task that attempts to find
  an element from a set. 
\begin{itemize}
\item[(a)]
 \emph{As per} the notation used in Sect.~\ref{sect:intro},
 an automated testing problem of program \FOO is a search
  problem $(\myX,\myU)$ where $\myX=\{x \mid \FOO(x)\failure\}$ and
  $\myU=\dom{\FOO}$. 
\item[(b)] Another search problem is satisfiability checking, where
  $\myX$ is the set of the models of a constraint, and $\myU$ is the
  value domain to which the variables can be assigned.
\end{itemize}
\label{eg:theory:sat}
\end{example}

\subsection{Representing Function}
\begin{definition}
  A function $\repf$ is said to be a \emph{representing function} for
  the search problem $(\myX,\myU)$ if with any $x\in\myU$ there is
  an associated real value $\repf(x)$, such that
  \begin{itemize}
  \item[(a)] $\repf(x)\geq 0$;
  \item[(b)] every root of the representation function is a
    solution of the search problem,  \ie,
    $\repf(x)=0 \implies x\in\myX$; and
  \item[(c)]  the roots of $\repf$ include all solutions to the
    search problem, \ie,  $x\in \myX \implies \repf(x)=0$.
  \end{itemize}

\label{def:theory:rp}
\end{definition}

\begin{example} Let $(\myX,\myU)$ be a search problem. 
  \begin{itemize}
  \item[(a)] 
   A trivial representing function is
  $\lambda x. (x\in\myX)?~ 0:1$.

\item[(b)] A generic representing function is the point-set distance.
  Imagine that the search problem is embedded in a metric
  space~\cite{Rudin76} with a distance
  $\mathnormal{dist}:\myX\times\myX\rightarrow \Real$.  As a standard
  practice, we can lift $\mathnormal{dist}$ to ${\mathnormal{dist}_X}$
  defined as
  $\lambda x.\inf\{\mathnormal{dist}(x,x')\mid x'\in \myX\}$, where
  $\inf$ refers to the greatest lower bound, or infimum.  Intuitively,
  ${\mathnormal{dist}_X}$ measures the distance between a point
  $x\in \myU$ and the set $\myX$. It can be shown that
  ${\mathnormal{dist}_X}$ satisfies conditions
  Def.~\ref{def:theory:rp}(a-c), and therefore, is a representing
  function.

\item[(c)] 
  The representing function used in branch coverage based testing is the \FOOR
  constructed in Sect.~\ref{sect:overview}, where $\myX$ is the input
  that saturates a new branch and $\myU$ is the input domain. We have
  proved that \FOOR is a representing function in
  Lem.~\ref{lem:overview:c2}.

  \end{itemize}

\end{example}

The theorem below allows us to approach a search problem by minimizing
its representing function.
\begin{thm}
Let $\repf$ be the representing function for the search problem $(\myX,\myU)$, and  $\repf^*$ be the global minimum  of $\repf$.
\begin{itemize}
\item[(a)] 
 Deciding the emptiness of $\myX$  is equivalent to checking the sign of $\repf^*$, \ie,  
$\myX=\emptyset \Leftrightarrow \repf^*>0$.
\item[(b)] 
Assume $\myX\neq\emptyset$. Then,
$\forall x\in\myU, x ~\text{minimizes}~\repf \Leftrightarrow x \in \myX$.
\end{itemize}
\label{thm:theory:rp}
\end{thm}
\begin{proof}
 Proof of (a):   Suppose $X\neq \emptyset$. Let $x_0$ be an element of $X$.  We have $\repf^*\geq 0$ by Def.~\ref{def:theory:rp}(a). In addition, we have $\repf^* \leq \repf(x_0)$ since $\repf^*$ is the minimum. Then we have $\repf^*\leq 0$ because $\repf(x_0)=0$ due to Def.~\ref{def:theory:rp}(c). Thus  $\repf^*=0$.   Conversely,  $\repf^*=0$ implies that there exists an $x^*\in \myU$ \st $\repf(x^*)=0$. By Def.~\ref{def:theory:rp}(b), $x^*\in \myX$. Thus $\myX\neq \emptyset$. 

 Proof of (b): Let $0_{\repf}$ denote the set of the roots of $\repf$, and
 $\myM_{\repf}$ be the set of the minimum points of $\repf$. It
 suffices to show that
 $0_{\repf}\subseteq \myM_{\repf}\subseteq \myX \subseteq 0_{\repf}$
 under the condition $X\neq \emptyset$.  We have
 $0_{\repf}\subseteq \myM_{\repf}$ because a root of $\repf$ is
 necessarily a minimum point;  we have $\myM_{\repf}\subseteq X$
 because $X\neq \emptyset$ implies $\repf^*=0$ by Thm.~\ref{thm:theory:rp}(a). Take an
 arbitrary $x^* \in \myM_{\repf}$, $\repf(x^*)=\repf^*=0$ holds, and
 therefore $x^*\in X$ by Def.~\ref{def:theory:rp}(b);  we have $X\subseteq 0_{\repf}$
 from Def.~\ref{def:theory:rp}(c).

\end{proof}

\begin{remark}
  An instance of Thm.~\ref{thm:theory:rp}(a) is shown in
  Tab.~\ref{tab:algo:example_illustration}, Line 4. There, \FOOR
  attains the minimum $1$, which means that all branches have been
  saturated (namely $X=\emptyset$).  An instance of
  Thm.~\ref{thm:theory:rp}(b) is shown in
  Eq.~\eqref{eq:intro:me}. Note that Thm.~\ref{thm:theory:rp}(b) does
  not hold if we drop the assumption $X\neq \emptyset$.  In fact, any
  $\repf$ such as $\repf(x)>0$ is a representing function for a search
  problem $(\emptyset, \myU)$, but its minimum point, if any, can
  never be an element of the empty set $\myX$.
\end{remark}

\subsection{The ME procedure}

The representing function paves the way toward a generic solution to
the search problem.  The key parameter in the \ME procedure is
mathematical optimization.

\begin{definition} Let $\myU$ be a set, $\mu$ be a mathematical
  optimization algorithm that attempts to calculate a minimum point for an 
  objective function defined over $\myU$. We define the \emph{Mathematical
    Execution}  procedure as follows:

\begin{center}
\mbox{
  \parbox{.96\linewidth}{
\vspace*{0.5pt}
\noindent Input: A search problem $(\myX,\myU)$ \\
\noindent Output: An element of $\myX$, or ``not found''
\vspace*{-1.5pt}
\begin{description}[align=left]
\item[{M1.}] Construct the representing function $\repf$. 
\item[{M2.}] Minimize  $\repf$.  Let  ${x^*}$ be the minimum point obtained by $\mu$.
\item[{M3.}] Return $x^*$ if $x^*\in X$.  Otherwise return  ``not found''.
\end{description}
  }
}
 \end{center}
\end{definition}







The following corollary states that the $\ME$ procedure solves the search
problem, under the condition that the $\ME$ procedure is equipped with an ideal mathematical
optimization backend. Its proof (omitted) follows Thm.~\ref{thm:theory:rp}.
\begin{cor}
  Let $(\myX,\myU)$ be a search problem.  Assume that $\mu$ yields a true global minimum point in  M2 of the $\ME$ procedure.  Then, 
 \begin{enumerate}
 \item[(a)] the $\ME$ procedure returns an
   $x^*\in X$ if $X\neq \emptyset$,  and
 \item[(b)] the $\ME$ procedure returns ``not
   found'' if $X=\emptyset$.
 \end{enumerate}
\label{cor:theory:theoretical}
  \end{cor}

\begin{remark}
  A known issue with testing is its incompleteness, \ie, ``Testing shows
  the presence of bugs, but not its absence''~\cite{EWD:EWD249}. In
  the context of Mathematical Execution, the above well-known remark corresponds
  to the fact that, in practice, if the $\ME$ procedure returns ``not found'' for the search
problem $(\myX,\myU)$, $\myX$ may still be non-empty.  We call this phenomenon \emph{practical incompleteness}.

The practical incompleteness occurs when the MO backend fails to yield
an accurate minimum point.  To clarify, we use $\underline{x^*}$ for
an exact global minimum point, and we use $x^*$ for the one calculated
by the MO backend $\mu$. Then we consider four disjunctive cases: 
\begin{itemize}
\item[(a)] $\repf({x^*})=0$ and $\repf(\underline{x^*})=0$;  
\item[(b)] $\repf({x^*})>0$ and $\repf(\underline{x^*})>0$; 
\item[(c)] $\repf({x^*})=0$ and $\repf(\underline{x^*})>0$;  
\item[(d)] $\repf({x^*})>0$ and $\repf(\underline{x^*})=0$. 
\end{itemize}
The \ME procedure remains correct for both (a) and (b). The case (c)
cannot happen because $\repf(\underline{x^*})<\repf(x)$ for all
$x$. The practical incompleteness occurs in (d), where the $\ME$
procedure returns ``not found'' but $X\neq\emptyset$.  Sect.~\ref{sect:discussion}
further discusses this incompleteness.
\label{remark:theory:practical}
\end{remark}

\subsection{Additional Examples}
\label{sect:theory:beyond}

This subsection aims to show that $\ME$ is a unified approach by
applying it on several other important search problems besides 
coverage-based testing. In each example, we illustrate $\ME$ with a
different representing function. 


\subsubsection{Path Reachability Testing}
\label{sect:overview:path}

Given a path $\tau$ of program $\FOO$, we call \emph{path reachability
  testing} the search problem $(\myX,\myU)$ with
\begin{align}
\myX=\{x\mid x \text{ triggers the path } \tau\}, \myU=\dom{\FOO}.
\end{align}
The path reachability problem has been studied as an independent
research topic~\cite{Miller:1976:AGF:1313320.1313530}, or more
commonly, as a subproblem in other testing
problems~\cite{DBLP:conf/pldi/GodefroidKS05,Korel:1990:AST:101747.101755,Lakhotia:2010:FSF:1928028.1928039}.

\paragraph{}
Consider the program $\FOO$ in Fig.~\ref{fig:overview:ME_illustration}(left), which is the same as that in
Fig.~\ref{fig:algo:injecting_pen}. Suppose that we want to trigger the
path $\tau=[0\myT,1\myT]$ (we denote a path by a sequence of
branches).  Our approach to this example problem is similar to the
three steps explained in Sect.~\ref{sect:overview:example}, except
that we design a different representing function here. We illustrate
the $\ME$ approach in Fig.~\ref{fig:overview:ME_illustration}.

\Paragraph{Step 1.}
We inject a global variable $\myr$ in $\FOO$ and the 
assignment 
\begin{align}
 \myr = \myr + d(op, a,b)
\end{align}
before the conditional statements, where $d$ is the
branch distance defined in Def.~\ref{def:overview:bd}.  The
instrumented program is shown as $\FOOI$ in
Fig.~\ref{fig:overview:ME_illustration}.   The assignment is to measure how
far the input has attained the desired path.



\Paragraph{Step 2.}
The value of $\myr$ is then retrieved through a driver program \FOOR
(Fig.~\ref{fig:overview:ME_illustration}), which initializes $\myr$ to
$0$ (unlike in Sect.~\ref{sect:overview:example}, where $\myr$ is initialized to $1$), calls \FOOI  and then returns $\myr$.  




\Paragraph{Step 3.}
A global minimum point of $\FOOR$ is calculated by an MO algorithm.
As shown in the graph of $\FOOR$ (Fig.~\ref{fig:overview:ME_illustration}), there are three local minimum
points, $\{-3,1,2\}$. Two of them, $\{-3,1\}$, attain the global
minimum $0$ and either solves the path reachability testing
problem. Note that the representing function is discontinuous at
$x=1$, but the MCMC procedure can easily find it (similar to  Fig.~\ref{fig:background:1}(b)).

Below we prove the correctness of the \ME solution. 
\begin{cor}
An input $x$ triggers $\tau$ iff $x$ minimizes $\FOOR$. 
\label{cor:theory:path_reachability2}
\end{cor}
\begin{proof}
  First, it can be shown that the constructed \FOOR is a
  representing function for the path reachability problem, thus we can
  apply Thm.~\ref{thm:theory:rp}.  By
  Thm.~\ref{thm:theory:rp}(a), $X\neq\emptyset$ since \FOOR attains
  the global minimum $0$. Then we conclude from Thm.~\ref{thm:theory:rp}(b).
\end{proof}




\subsubsection{Boundary Value Analysis}\label{sect:overview:bva}

\begin{figure}
  \includegraphics[width=1.0\linewidth]{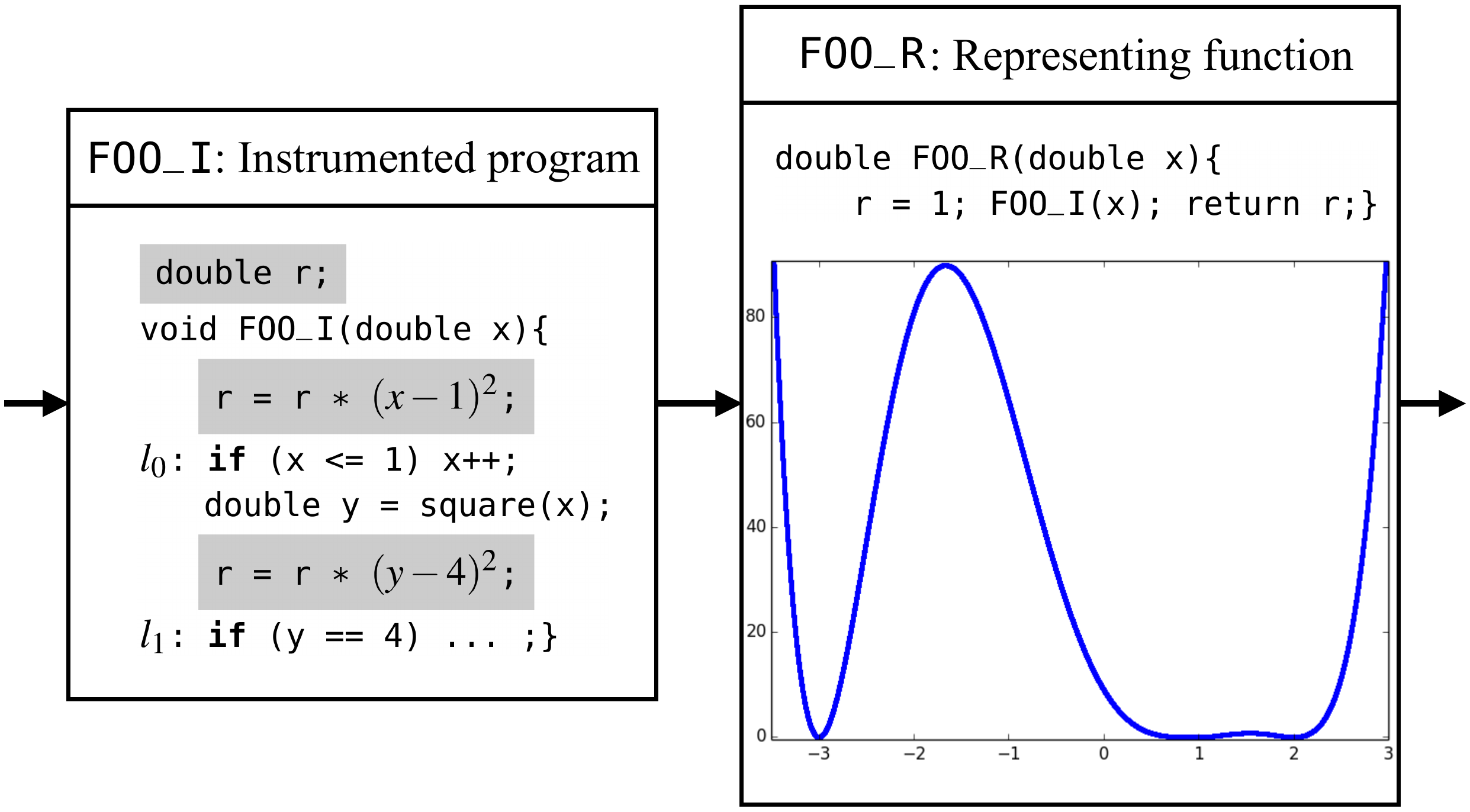}
  \caption{Boundary value analysis via $\ME$. This goal  is to find
    a test input to trigger a boundary condition, namely,  (a) $x=1$ at $l_0$ or (b) $y=4$ at $l_1$ of the program 
    $\FOO$ in Fig.~\ref{fig:overview:ME_illustration}.}
\label{fig:overview:bva}
\end{figure}
In testing, test inputs that explore ``boundary conditions'' usually have a higher
payoff than those that do
not~\cite{White:1980:DSC:1313339.1313789}. The problem of generating
such inputs is expressed abstractly as \emph{boundary value
  analysis}~\cite{DBLP:dblp_books/daglib/0012071,%
  Kosmatov:2004:BCC:1032654.1033824,DBLP:dblp_conf/icsm/PanditaXTH10},
which can be seen as the search problem with
\begin{align}
\myX=\{x\mid x \text{ triggers a boundary condition}\}, \myU=\dom{\FOO}.
\end{align}

Consider again the program $\FOO$ in
Fig.~\ref{fig:overview:ME_illustration}. The boundary value analysis
is to find test inputs that trigger a boundary condition (a) $x=1$ at
$l_0$, or (b) $y=4$ at $l_1$.  With manual reasoning, condition (a)
can only be triggered by input $1$. Condition (b) can be triggered if
$x=2$ or $x= -2$ before $l_1$. For each case, we reason backward
across the two branches at $l_0$ and then merge the results.
Eventually we have $\{-3,1,2\}$ as the test inputs to generate.






\paragraph{}
Our $\ME$ solution follows: As before, we introduce a global variable $\myr$ to estimate how far a program 
input is from triggering a boundary condition.  We inject the
assignment
\begin{align}
\myr = \myr * d(==, a,b)
\end{align}
before each condition $a~op~b$, where function $d$ is the branch
distance defined in Def.~\ref{def:overview:bd}.
Fig.~\ref{fig:overview:bva} illustrates $\FOOI$ and $\FOOR$. Then we
can solve the boundary value analysis problem via minimizing \FOOR.
The correctness of this procedure follows:
\begin{cor} 
  An input $x$ triggers a boundary condition if and only if $x$
  minimizes \FOOR.
\label{lem:overview:bva}
\end{cor}
\begin{proof}
  It can be shown that (1) $\myr$ is always assigned to a non-negative
  value; (2) if $\myr=0$ at the end of the program execution, then one
  of the boundary conditions must hold; and (3) $\myr=0$ at the end of
  the program execution if at least one of the boundary conditions
  holds.  Thus, the constructed \FOOR satisfies the conditions for
  being a representing function (Def.~\ref{def:theory:rp}). We can
  then prove the corollary following Thm.~\ref{thm:theory:rp} (in the
  same way as we prove Cor.~\ref{cor:theory:path_reachability2}).
\end{proof}


\subsubsection{Satisfiability Checking}
\label{sect:overview:sat}

 Consider the  satisfiability problem with constraint
\begin{align}\label{eq:theory:sat2}
\pi=2^x\leq 5 \land x^2\geq 5\land x\geq 0
\end{align}
where $x\in\Real$.  If we write $x^*\models \pi$ to mean that  $\pi$ becomes a tautology by substituting its free variable $x$ with $x^*$.  Then, this satisfiability  problem can be expressed as the search problem 
$(\{x^*\mid x^*\models \pi\}, \Real)$.
Its representing function can be defined as (Fig.~\ref{fig:theory:sat}): 
\begin{align}\label{eq:theory:sat}
  \begin{split}
    \repf\mydef\lambda x. (2^x\leq 5)~?~ 0 : (2^x - 5)^2 &~+~(x^2\geq 5)~?~ 0 : (x^2-5)^2 \\
    &~+~ (x\geq 0)~?~0:x^2.
  \end{split}
\end{align}
Then we can solve the satisfiability checking problem of constraint
$\pi$ in Eq.~\eqref{eq:theory:sat2} by minimizing its representing
function $\repf$ in Eq.~\eqref{eq:theory:sat}.  The \ME procedure can
easily locate the minimum points in between $\sqrt{5}$
($\approx 2.24$) and $\log_2 5$ ($\approx 2.32$). Each of the minimum
points is necessarily a model of $\pi$ following
Thm.~\ref{thm:theory:rp}. The correctness of such an approach follows:

\begin{cor}
$x^*\models \pi \Leftrightarrow x^* \text{ minimizes } \repf$.
\end{cor}

\begin{figure}[!t] \centering
  \includegraphics[width=0.6\linewidth]{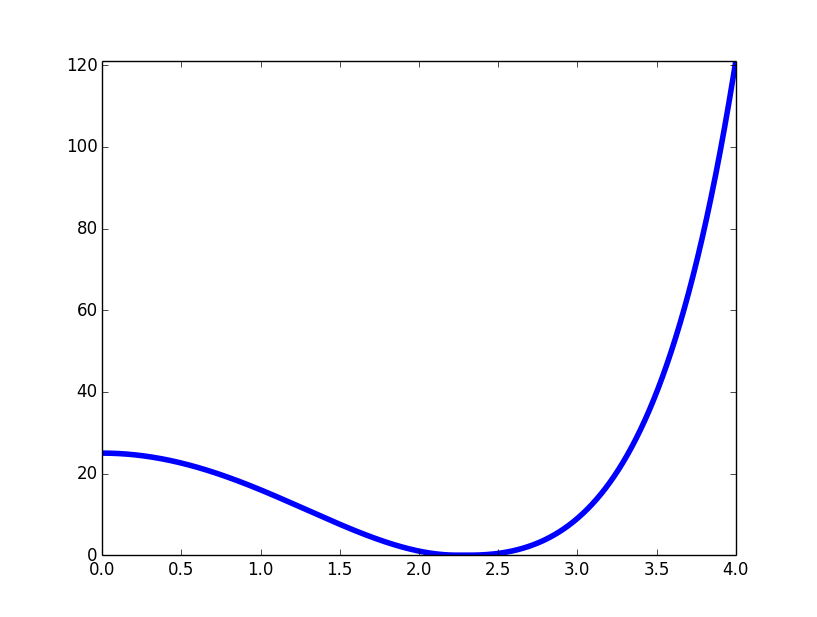}
  \caption{Satisfiability checking via $\ME$. This goal of this example is to find a model of 
$\pi$ defined in Eq.~\eqref{eq:theory:sat2}.
The curve depicts its representing function defined in Eq.~\eqref{eq:theory:sat}.
     }\label{fig:theory:sat}
\end{figure}

\begin{remark}
  The representing function in Eq.~\eqref{eq:theory:sat} is defined
  on $\Real$. It is to illustrate the concept and allows us to ignore
  issues about floating-point inaccuracy. A more realistic
  representing function for satisfiability checking has been shown in
  a recent work~\cite{xsat} where the 
  representing function is constructed based on the binary forms of floating-point
  numbers to avoid floating-point arithmetic in the first place.
\end{remark}







\section{Evaluation}
\label{sect:eval}



\subsection{The ME-powered System CoverMe}

We have implemented CoverMe, a proof-of-concept realization for branch
coverage based testing.  \coverme has a modular design and can be
adapted to other automated testing problems, such as path reachability
and boundary value analysis (Sect.~\ref{sect:theory:beyond}).  This
section presents the architecture of \coverme (Fig. \ref{fig:1103}),
in which we identify three layers and their corresponding use cases.
Implementation details are given in Appendix~\ref{sect:implem}.


\Paragraph{Client.} This layer
provides the program under test $\FOO$ and specifies the inputs to
generate $X$. The search problem involved will be
$(X,\dom{\FOO})$. $X$ is usually implicit, \eg, specified via a path
to trigger as in path reachability testing
(Sect.~\ref{sect:overview:path}), and checking $x\in X$ has to be
feasible following Def.~\ref{def:theory:search}.  $\FOO$ is a piece of
code in the LLVM intermediate
representation~\cite{Lattner:2004:LCF:977395.977673} or in any
language that can be transformed to it (such as code in Ada, the C/C++
language family, or Julia). The current \coverme  has been  tested on C, and we require $\dom{\FOO}\subseteq \Real^N$, as the system backend outputs floating-point test data by default (see the ME kernel layer below). 




\Paragraph{Researcher.} This
layer sets two parameters. One is the initial
value of the representing function, denoted by $\myr_0$. It is usually
either $0$ or $1$ from our experience. For example, $\myr_0$ is set to $0$ in path
reachability testing (Sect.~\ref{sect:overview:path}), and $1$ in both coverage-based testing (Sect.~\ref{sect:overview:example}) and
boundary value analysis (Sect.~\ref{sect:overview:bva}). The other parameter is the assignment to
inject before each conditional statement in $\FOO$.  In practice, the
Researcher specifies this code segment in the $\pen$ procedure and
injects $\myr=\pen$, as shown in Sect.~\ref{sect:overview}.







\Paragraph{ME Kernel.} This layer takes as inputs the program $\FOO$ provided by
the Client, $\myr_0$ and $\pen$ set by the Researcher, and operates
the three steps described in Sect.~\ref{sect:overview:example}. It (1) uses
Clang~\cite{clang:web} to inject $\myr$ and $\pen$ into $\FOO$ to
construct $\FOOI$, (2) initializes $\myr$ to $\myr_0$, invokes $\FOOI$
and returns $\myr$ in $\FOOR$, and (3) minimizes $\FOOR$ with
MO. 
As mentioned in Sect.~\ref{sect:theory:algo}, \coverme uses Basinhopping, an off-the-shelf implementation from SciPy optimization package~\cite{scipy:web} as the MO backend.  Basinhopping is then launched from  different starting points  as shown in Algo.~\ref{theory:algo:coverme}, Lines 8-12.  These starting points are randomly generated from the Hypothesis library~\cite{hypothesis:web}.

\begin{figure}[t]
  \centering
\includegraphics[width=1.0\linewidth]{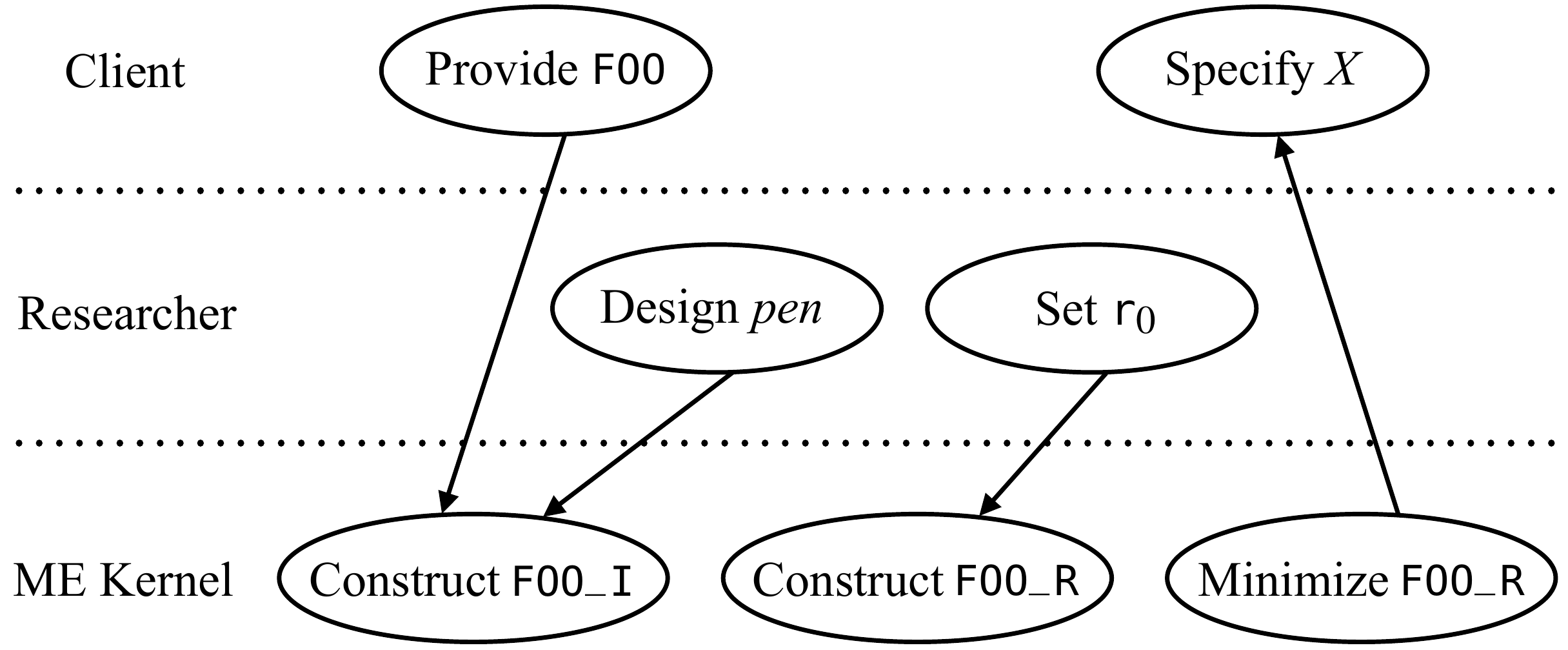}
\caption{Architecture of $\coverme$. Each layer is associated with  typical use cases. $\FOO$: program under
  test; $\myX$:  test inputs to generate;  $\pen$: procedure
  that updates $\myr$; $\myr_0$: initial value of the representing
  function; $\FOOI$: instrumented program; $\FOOR$:  representing function. }
\label{fig:1103}
\end{figure}

\subsection{Experimental Setup}
\label{subsect:expsetup}
This section discusses the benchmarks, the tools for comparison, the settings and hardware, and two evaluation objectives.

\Paragraph{Benchmarks.}
We use the  C math library Fdlibm
5.3~\cite{fdlibm:web} as our benchmark. These  programs
are developed by Sun (now Oracle). They are real-world programs, rich in floating-point operations, and have been used in  Matlab, Java, JavaScript and Android.

Fdlibm includes $80$ programs.  Each has one or multiple entry
functions. In total, Fdlibm has $92$ entry functions. Among them, we
exclude (1) $36$ functions that do not have  branches, (2) $11$ functions
involving non-floating-point input parameters, and (3) $5$ static C
functions. Our benchmark suite includes
\emph{all} remaining $40$ functions in Fdlibm. 
For completeness,
 we list untested functions and reasons why
they are not selected in Appendix~\ref{sect:untested}.



\Paragraph{Compared  Tools.}
We have considered  tools  publicly available to us, including Austin~\cite{lakhotia2013austin}, Pex~\cite{Tillmann:2008:PWB:1792786.1792798}, and 
Klee~\cite{Cadar:2008:KUA:1855741.1855756} and its two variants, namely  Klee-Mulitsolver~\cite{Palikareva:2013:MSS:2526861.2526865}  and
Klee-FP~\cite{collingbourne2011symbolic}.
We tried Klee-Mulitsolver~\cite{Palikareva:2013:MSS:2526861.2526865} with
Z3~\cite{DeMoura:2008:ZES:1792734.1792766} as the SMT backend but found
that the expression language of
Klee~\cite{Cadar:2008:KUA:1855741.1855756} did not support
floating-point constraints. Besides, some common operations in
our benchmark programs, such as pointer reference and dereference,
type casting, external function calls, \etc, are not supported by Z3 
or any other backend solvers compatible with  Klee-Multisolver. Klee-FP supports symbolic reasoning on the equivalence between
floating-point values but does not support coverage testing~\cite{collingbourne2011symbolic}. Pex,
unfortunately, can only run for .NET programs on Windows whereas
Fdlibm is in  C, and our testing platform is based on Unix.

To our best knowledge, Austin~\cite{lakhotia2013austin} is the only
publicly available tool supporting coverage testing for Fdlibm. Austin
combines symbolic execution and search-based heuristics, and has
been thoroughly tested for branch coverage based testing on
floating-point programs~\cite{lakhotia2010empirical}.  For empirical comparison,
we have also implemented a random sampling tool, which samples inputs from the
function's input domains using a standard pseudo-random number
generator. We refer to the tool as Rand.


\Paragraph{Settings and Hardware.}
CoverMe supports the following command line options while generating
test inputs from floating-point programs: (1) the number of
Monte-Carlo iterations $\niter$, (2) the local optimization algorithm
\LM and (3) the number of starting points $\nStart$.  These options
correspond to the three input parameters in
Algo.~\ref{theory:algo:coverme}.  We set $\niter=5$, $\nStart=500$,
and $\LM$=``powell'' which refers to Powell’s local optimization
algorithm~\cite{Press:2007:NRE:1403886}.

For both Austin and CoverMe, we use the default settings for running the benchmarks. 
All experiments were performed on a laptop with a 2.6 GHz Intel Core i7 and 4GB RAM running Ubuntu 14.04
virtual machine.

\Paragraph{Evaluation Objectives.}

\setlength{\tabcolsep}{2pt}
\begin{table} \footnotesize
	\caption {Gcov metrics and explanations}\label{tab:algo:gcov}
	\begin{tabular}{p{0.625in} p{0.7in} p{1in} p{0.7in}}
		\toprule
		Metrics &  \mbox{Gcov} message & Description & Note \\
		\midrule
		Line\% & Lines executed & Covered source lines over total source lines & \mbox{\aka line or} \mbox{statement coverage} \\\arrayrulecolor{lightgray}\hline
		Condition\% & Branches \mbox{executed} & \mbox{Covered conditional} \mbox{statements over total} conditional statements \\\arrayrulecolor{lightgray}\hline
		Branch\% & Branches taken at least once & Covered branches over total branches & \aka \mbox{branch coverage} \\\arrayrulecolor{lightgray}\hline
		Call\% & Calls executed & \mbox{Covered calls over} \mbox{total} calls  \\
		\bottomrule
	\end{tabular}
\end{table}

There are two specific evaluation objectives. (1) \emph{Coverage}: We use
the standard Gnu coverage tool Gcov~\cite{gcov:web} to analyze the
coverage. Gcov generates four metrics for source code coverage
analysis, which are listed as ``Lines'', ``Conditions'', ``Branches''
and ``Calls'' in Col. 1, Tab.~\ref{tab:algo:gcov}.  Col.~2-4 give the
corresponding Gcov report message, the descriptions and the general
metrics for coverage testing.  (2) \emph{Efficiency}: we measure the wall
time reported by the standard Unix command ``time''.  The timeout
limit for the compared tools is 48 hours.





\subsection{Quantitative Results}

This subsection presents two sets of experimental results. The first validates our approach by comparing it against random
testing. 
The second compares CoverMe with Austin~\cite{lakhotia2013austin}, an open-source tool that combines symbolic execution and search-based strategies.

\subsubsection{CoverMe versus Random Testing}
\setlength{\tabcolsep}{1pt}
\begin{table*}[th]\centering \footnotesize             
  \caption {Comparing Random testing and CoverMe.  The benchmark programs are taken from Fdlibm~\cite{fdlibm:web}. We calculate the coverage percentage using Gcov~\cite{gcov:web}. The metrics of Gcov are explained in Tab.~\ref{tab:algo:gcov}.  
    ``n/a''  in the last two columns indicates that no function call exists in the
    program; lines with ``n/a'' are excluded when calculating the  mean values of the last two columns.
  } \label{tab:eval:me_vs_random}
	\begin{tabular}{llrrr|r|rr|rr|rr|rr}
		\toprule
		\multicolumn{5}{c|}{Benchmark characteristics}&\multicolumn{1}{c|}{Time (s)}&\multicolumn{2}{c|}{Line  (\%)} &\multicolumn{2}{c|}{Condition  (\%)} &\multicolumn{2}{c|}{ Branche  (\%)}&\multicolumn{2}{c}{Call  (\%)}
		\\\arrayrulecolor{lightgray}\hline
		Program & Entry function & \#Line  & \#Branch  & \#Call &  & Rand & CoverMe &  Rand & CoverMe & Rand &  CoverMe  & Rand & CoverMe 
		\\\arrayrulecolor{black}\midrule
		e\_acos.c & ieee754\_acos(double) & 33 & 12 & 0 & 7.8 & 18.2 & 100.0 & 33.3 & 100.0 & 16.7 & 100.0 & n/a & n/a \\\arrayrulecolor{lightgray}\hline
		e\_acosh.c & ieee754\_acosh(double) & 15 & 10 & 2 & 2.3 & 46.7 & 93.3 & 60.0 & 100.0 & 40.0 & 90.0 & 50.0 & 100.0 \\\arrayrulecolor{lightgray}\hline
		e\_asin.c & ieee754\_asin(double) & 31 & 14 & 0 & 8.0 & 19.4 & 100.0 & 28.6 & 100.0 & 14.3 & 92.9 & n/a & n/a \\\arrayrulecolor{lightgray}\hline
		e\_atan2.c & ieee754\_atan2(double, double) & 39 & 44 & 0 & 17.4 & 59.0 & 79.5 & 54.6 & 84.1 & 34.1 & 63.6 & n/a & n/a \\\arrayrulecolor{lightgray}\hline
		e\_atanh.c & ieee754\_atanh(double) & 15 & 12 & 0 & 8.1 & 40.0 & 100.0 & 16.7 & 100.0 & 8.8 & 91.7 & n/a & n/a \\\arrayrulecolor{lightgray}\hline
		e\_cosh.c & ieee754\_cosh(double) & 20 & 16 & 3 & 8.2 & 50.0 & 100.0 & 75.0 & 100.0 & 37.5 & 93.8 & 0.0 & 100.0 \\\arrayrulecolor{lightgray}\hline
		e\_exp.c & ieee754\_exp(double) & 31 & 24 & 0 & 8.4 & 25.8 & 96.8 & 33.3 & 100.0 & 20.8 & 96.7 & n/a & n/a \\\arrayrulecolor{lightgray}\hline
		e\_fmod.c & ieee754\_frmod(double, double) & 70 & 60 & 0 & 22.1 & 54.3 & 77.1 & 66.7 & 80.0 & 48.3 & 70.0 & n/a & n/a \\\arrayrulecolor{lightgray}\hline
		e\_hypot.c & ieee754\_hypot(double, double) & 50 & 22 & 0 & 15.6 & 66.0 & 100.0 & 63.6 & 100.0 & 40.9 & 90.9 & n/a & n/a \\\arrayrulecolor{lightgray}\hline
		e\_j0.c & ieee754\_j0(double) & 29 & 18 & 2 & 9.0 & 55.2 & 100.0 & 55.6 & 100.0 & 33.3 & 94.4 & 0.0 & 100.0 \\\arrayrulecolor{lightgray}\hline
		& ieee754\_y0(double) & 26 & 16 & 5 & 0.7 & 69.2 & 100.0 & 87.5 & 100.0 & 56.3 & 100.0 & 0.0 & 100.0 \\\arrayrulecolor{lightgray}\hline
		e\_j1.c & ieee754\_j1(double) & 26 & 16 & 2 & 10.2 & 65.4 & 100.0 & 75.0 & 100.0 & 50.0 & 93.8 & 0.0 & 100.0 \\\arrayrulecolor{lightgray}\hline
		& ieee754\_y1(double) & 26 & 16 & 4 & 0.7 & 69.2 & 100.0 & 87.5 & 100.0 & 56.3 & 100.0 & 0.0 & 100.0 \\\arrayrulecolor{lightgray}\hline
		e\_log.c & ieee754\_log(double) & 39 & 22 & 0 & 3.4 & 87.7 & 100.0 & 90.9 & 100.0 & 59.1 & 90.9 & n/a & n/a \\\arrayrulecolor{lightgray}\hline
		e\_log10.c & ieee754\_log10(double) & 18 & 8 & 1 & 1.1 & 83.3 & 100.0 & 100.0 & 100.0 & 62.5 & 87.5 & 100.0 & 100.0 \\\arrayrulecolor{lightgray}\hline
		e\_pow.c & ieee754\_pow(double, double) & 139 & 114 & 0 & 18.8 & 15.8 & 92.7 & 28.1 & 92.7 & 15.8 & 81.6 & n/a & n/a \\\arrayrulecolor{lightgray}\hline
		e\_rem\_pio2.c & ieee754\_rem\_pio2(double, double*) & 64 & 30 & 1 & 1.1 & 29.7 & 92.2 & 46.7 & 100.0 & 33.3 & 93.3 & 100.0 & 100.0 \\\arrayrulecolor{lightgray}\hline
		e\_remainder.c & ieee754\_remainder(double, double) & 27 & 22 & 1 & 2.2 & 77.8 & 100.0 & 72.7 & 100.0 & 45.5 & 100.0 & 100.0 & 100.0 \\\arrayrulecolor{lightgray}\hline
		e\_scalb.c & ieee754\_scalb(double, double) & 9 & 14 & 0 & 8.5 & 66.7 & 100.0 & 85.7 & 100.0 & 50.0 & 92.9 & n/a & n/a \\\arrayrulecolor{lightgray}\hline
		e\_sinh.c & ieee754\_sinh(double) & 19 & 20 & 2 & 0.6 & 57.9 & 100.0 & 60.0 & 100.0 & 35.0 & 95.0 & 0.0 & 100.0 \\\arrayrulecolor{lightgray}\hline
		e\_sqrt.c & ieee754\_sqrt(double) & 68 & 46 & 0 & 15.6 & 85.3 & 94.1 & 87.0 & 92.7 & 69.6 & 82.6 & n/a & n/a \\\arrayrulecolor{lightgray}\hline
		k\_cos.c & kernel\_cos(double, double) & 15 & 8 & 0 & 15.4 & 73.3 & 100.0 & 75.0 & 100.0 & 37.5 & 87.5 & n/a & n/a \\\arrayrulecolor{lightgray}\hline
		s\_asinh.c & asinh(double) & 14 & 12 & 2 & 8.4 & 57.1 & 100.0 & 66.7 & 100.0 & 41.7 & 91.7 & 50.0 & 100.0 \\\arrayrulecolor{lightgray}\hline
		s\_atan.c & atan(double) & 28 & 26 & 0 & 8.5 & 25.0 & 96.4 & 30.8 & 100.0 & 19.2 & 88.5 & n/a & n/a \\\arrayrulecolor{lightgray}\hline
		s\_cbrt.c & cbrt(double) & 24 & 6 & 0 & 0.4 & 87.5 & 91.7 & 100.0 & 100.0 & 50.0 & 83.3 & n/a & n/a \\\arrayrulecolor{lightgray}\hline
		s\_ceil.c & ceil(double) & 29 & 30 & 0 & 8.8 & 27.6 & 100.0 & 20.0 & 100.0 & 10.0 & 83.3 & n/a & n/a \\\arrayrulecolor{lightgray}\hline
		s\_cos.c & cos (double) & 12 & 8 & 6 & 0.4 & 100.0 & 100.0 & 100.0 & 100.0 & 75.0 & 100.0 & 83.3 & 100.0 \\\arrayrulecolor{lightgray}\hline
		s\_erf.c & erf(double) & 38 & 20 & 2 & 9.0 & 21.1 & 100.0 & 50.0 & 100.0 & 30.0 & 100.0 & 0.0 & 100.0 \\\arrayrulecolor{lightgray}\hline
		& erfc(double) & 43 & 24 & 2 & 0.1 & 18.6 & 100.0 & 41.7 & 100.0 & 25.0 & 100.0 & 0.0 & 100.0 \\\arrayrulecolor{lightgray}\hline
		s\_expm1.c & expm1(double) & 56 & 42 & 0 & 1.1 & 21.4 & 100.0 & 33.3 & 100.0 & 21.4 & 97.6 & n/a & n/a \\\arrayrulecolor{lightgray}\hline
		s\_floor.c & floor(double) & 30 & 30 & 0 & 10.1 & 26.7 & 100.0 & 20.0 & 100.0 & 10.0 & 83.3 & n/a & n/a \\\arrayrulecolor{lightgray}\hline
		s\_ilogb.c & ilogb(double) & 12 & 12 & 0 & 8.3 & 33.3 & 91.7 & 33.3 & 83.3 & 16.7 & 75.0 & n/a & n/a \\\arrayrulecolor{lightgray}\hline
		s\_log1p.c & log1p(double) & 46 & 36 & 0 & 9.9 & 71.7 & 100.0 & 61.1 & 100.0 & 38.9 & 88.9 & n/a & n/a \\\arrayrulecolor{lightgray}\hline
		s\_logb.c & logb(double) & 8 & 6 & 0 & 0.3 & 87.5 & 87.5 & 100.0 & 100.0 & 50.0 & 83.3 & n/a & n/a \\\arrayrulecolor{lightgray}\hline
		s\_modf.c & modf(double, double*) & 32 & 10 & 0 & 3.5 & 31.2 & 100.0 & 46.7 & 100.0 & 33.3 & 100.0 & n/a & n/a \\\arrayrulecolor{lightgray}\hline
		s\_nextafter.c & nextafter(double, double) & 36 & 44 & 0 & 17.5 & 72.2 & 88.9 & 81.8 & 95.5 & 59.1 & 79.6 & n/a & n/a \\\arrayrulecolor{lightgray}\hline
		s\_rint.c & rint(double) & 34 & 20 & 0 & 3.0 & 26.5 & 100.0 & 30.0 & 100.0 & 15.0 & 90.0 & n/a & n/a \\\arrayrulecolor{lightgray}\hline
		s\_sin.c & sin (double) & 12 & 8 & 6 & 0.3 & 100.0 & 100.0 & 100.0 & 100.0 & 75.0 & 100.0 & 83.3 & 100.0 \\\arrayrulecolor{lightgray}\hline
		s\_tan.c & tan(double) & 6 & 4 & 3 & 0.3 & 100.0 & 100.0 & 100.0 & 100.0 & 50.0 & 100.0 & 66.7 & 100.0 \\\arrayrulecolor{lightgray}\hline
		s\_tanh.c & tanh(double) & 16 & 12 & 0 & 0.7 & 43.8 & 100.0 & 50.0 & 100.0 & 33.3 & 100.0 & n/a & n/a \\\arrayrulecolor{black}\midrule
		MEAN &  &  &  &  & 6.9 & 54.2 & 97.0 & 61.2 & 98.2 & 38.0 & 90.8 & 39.6 & 100.0 \\\arrayrulecolor{black}\bottomrule
	\end{tabular}
\end{table*}

We have compared CoverMe with Rand by running them on the benchmarks
described in Sect.~\ref{subsect:expsetup}. In
Tab.~\ref{tab:eval:me_vs_random}, we sort all benchmark programs
(Col.~1) and entry functions (Col.~2) by their names, give the numbers
of source lines, branches and invoked functions (Col.~3-5). All
coverage results are given by Gcov. It can be seen that the programs
are rich in branches. The largest number of branches is 114 ({\tt
  ieee754\_pow}).  Especially, some tested functions have more
branches than lines. For example, {\tt ieee754\_atan2} has 44
branches, but only 39 lines; {\tt ceil} has 30 branches, but only 29
lines.


The times used by CoverMe are given by Col.~6 in
Tab.~\ref{tab:eval:me_vs_random}.  Observe that the times used by
CoverMe vary considerably through all entry functions, from 0.1 to
22.1 seconds, but all under half a minute. We find the numbers of
lines (Col.~3) and branches (Col.~4) are not correlated with the
running times (Col.~6). CoverMe takes 1.1 seconds to run the function
{\tt expm1} (with 42 branches and 56 lines) and 10.1 seconds to run
the function {\tt floor} (with 30 branches and 30 lines). It shows the
potential for real-world program testing since CoverMe may not be very
sensitive to the number of lines or branches.  We set the timeout limit
of Rand as 600 seconds, since Rand does not terminate by
itself and 600 seconds are already larger than the times spent by
CoverMe by orders of magnitude.


Col.~7-14 in Tab.~\ref{tab:eval:me_vs_random} show the coverage results of Rand and CoverMe.
We use Gcov to compute four metrics: lines, condition, branches and calls (as defined in Tab.~\ref{tab:algo:gcov}).  All coverage results reported
by Gcov have been included in Tab.~\ref{tab:eval:me_vs_random} for completeness. The columns ``Line'' and ``Branch'' refer to the commonly used line/statement coverage and branch coverage,
respectively.  The average values of the coverage are shown in the last row of the table. All values in Col.~8, 10, 12 and 14 are larger than or equal to the corresponding values in Col.~7, 9, 11 and 13. It means that CoverMe
achieves higher coverage  than Rand for every benchmark program and every metric.  For the line and condition coverages, CoverMe
achieves full coverage for almost all benchmarks, with 97.0\% line coverage and 
98.2\% condition coverage on average, whereas Rand achieves 54.2\%
and 61.2\% for the two metrics. For the most important branch coverage (since it is a primary target of this paper), CoverMe achieves 100\%
coverage for $11$ out of $40$ benchmarks with an average of 90.8\% coverage, while Rand does not achieve any 100\% coverage and attains only 38.0\% coverage on average. CoverMe achieves 100\% call coverage on average, whereas Rand achieves 39.5\%. 

\subsubsection{CoverMe versus Austin}

\setlength{\tabcolsep}{3.5pt}
\begin{table*} \centering\footnotesize
	\caption{Comparison of Austin~\cite{lakhotia2013austin}  and CoverMe.  The benchmark programs are taken from the Fdlibm library~\cite{fdlibm:web}.}\label{tab:eval:me_vs_austin}
	\begin{tabular}{ll|rr|rr|rr}
		\toprule
		\multicolumn{2}{c|}{Benchmark}&\multicolumn{2}{c|}{Time (second)}&\multicolumn{2}{c|}{Branch coverage(\%)} &\multicolumn{2}{c}{Improvement metrics}
		\\\arrayrulecolor{lightgray}\hline
		Program & Entry function  & Austin & CoverMe & Austin & CoverMe & Speedup &  Coverage (\%) \\\arrayrulecolor{black}\midrule
		e\_acos.c & ieee754\_acos(double) & 6058.8 & 7.8 & 16.7 & 100.0 & 776.4 & 83.3 \\\arrayrulecolor{lightgray}\hline
		e\_acosh.c & ieee754\_acosh(double) & 2016.4 & 2.3 & 40.0 & 90.0 & 887.5 & 50.0 \\\arrayrulecolor{lightgray}\hline
		e\_asin.c & ieee754\_asin(double) & 6935.6 & 8.0 & 14.3 & 92.9 & 867.0 & 78.6 \\\arrayrulecolor{lightgray}\hline
		e\_atan2.c & ieee754\_atan2(double, double) & 14456.0 & 17.4 & 34.1 & 63.6 & 831.2 & 29.6 \\\arrayrulecolor{lightgray}\hline
		e\_atanh.c & ieee754\_atanh(double) & 4033.8 & 8.1 & 8.3 & 91.7 & 495.4 & 83.3 \\\arrayrulecolor{lightgray}\hline
		e\_cosh.c & ieee754\_cosh(double) & 27334.5 & 8.2 & 37.5 & 93.8 & 3327.7 & 56.3 \\\arrayrulecolor{lightgray}\hline
		e\_exp.c & ieee754\_exp(double) & 2952.1 & 8.4 & 75.0 & 96.7 & 349.7 & 21.7 \\\arrayrulecolor{lightgray}\hline
		e\_fmod.c & ieee754\_frmod(double, double) & timeout & 22.1 & n/a & 70.0 & n/a & n/a \\\arrayrulecolor{lightgray}\hline
		e\_hypot.c & ieee754\_hypot(double, double) & 5456.8 & 15.6 & 36.4 & 90.9 & 350.9 & 54.6 \\\arrayrulecolor{lightgray}\hline
		e\_j0.c & ieee754\_j0(double) & 6973.0 & 9.0 & 33.3 & 94.4 & 776.5 & 61.1 \\\arrayrulecolor{lightgray}\hline
		& ieee754\_y0(double) & 5838.3 & 0.7 & 56.3 & 100.0 & 8243.5 & 43.8 \\\arrayrulecolor{lightgray}\hline
		e\_j1.c & ieee754\_j1(double) & 4131.6 & 10.2 & 50.0 & 93.8 & 403.9 & 43.8 \\\arrayrulecolor{lightgray}\hline
		& ieee754\_y1(double) & 5701.7 & 0.7 & 56.3 & 100.0 & 8411.0 & 43.8 \\\arrayrulecolor{lightgray}\hline
		e\_log.c & ieee754\_log(double) & 5109.0 & 3.4 & 59.1 & 90.9 & 1481.9 & 31.8 \\\arrayrulecolor{lightgray}\hline
		e\_log10.c & ieee754\_log10(double) & 1175.5 & 1.1 & 62.5 & 87.5 & 1061.3 & 25.0 \\\arrayrulecolor{lightgray}\hline
		e\_pow.c & ieee754\_pow(double, double) & timeout & 18.8 & n/a & 81.6 & n/a & n/a \\\arrayrulecolor{lightgray}\hline
		e\_rem\_pio2.c & ieee754\_rem\_pio2(double, double*) & timeout & 1.1 & n/a & 93.3 & n/a & n/a \\\arrayrulecolor{lightgray}\hline
		e\_remainder.c & ieee754\_remainder(double, double) & 4629.0 & 2.2 & 45.5 & 100.0 & 2146.5 & 54.6 \\\arrayrulecolor{lightgray}\hline
		e\_scalb.c & ieee754\_scalb(double, double) & 1989.8 & 8.5 & 57.1 & 92.9 & 233.8 & 35.7 \\\arrayrulecolor{lightgray}\hline
		e\_sinh.c & ieee754\_sinh(double) & 5534.8 & 0.6 & 35.0 & 95.0 & 9695.9 & 60.0 \\\arrayrulecolor{lightgray}\hline
		e\_sqrt.c & iddd754\_sqrt(double) & crash & 15.6 & n/a & 82.6 & n/a & n/a \\\arrayrulecolor{lightgray}\hline
		k\_cos.c & kernel\_cos(double, double) & 1885.1 & 15.4 & 37.5 & 87.5 & 122.6 & 50.0 \\\arrayrulecolor{lightgray}\hline
		s\_asinh.c & asinh(double) & 2439.1 & 8.4 & 41.7 & 91.7 & 290.8 & 50.0 \\\arrayrulecolor{lightgray}\hline
		s\_atan.c & atan(double) & 7584.7 & 8.5 & 26.9 & 88.5 & 890.6 & 61.6 \\\arrayrulecolor{lightgray}\hline
		s\_cbrt.c & cbrt(double) & 3583.4 & 0.4 & 50.0 & 83.3 & 9109.4 & 33.3 \\\arrayrulecolor{lightgray}\hline
		s\_ceil.c & ceil(double) & 7166.3 & 8.8 & 36.7 & 83.3 & 812.3 & 46.7 \\\arrayrulecolor{lightgray}\hline
		s\_cos.c & cos (double) & 669.4 & 0.4 & 75.0 & 100.0 & 1601.6 & 25.0 \\\arrayrulecolor{lightgray}\hline
		s\_erf.c & erf(double) & 28419.8 & 9.0 & 30.0 & 100.0 & 3166.8 & 70.0 \\\arrayrulecolor{lightgray}\hline
		& erfc(double) & 6611.8 & 0.1 & 25.0 & 100.0 & 62020.9 & 75.0 \\\arrayrulecolor{lightgray}\hline
		s\_expm1.c & expm1(double) & timeout & 1.1 & n/a & 97.6 & n/a & n/a \\\arrayrulecolor{lightgray}\hline
		s\_floor.c & floor(double) & 7620.6 & 10.1 & 36.7 & 83.3 & 757.8 & 46.7 \\\arrayrulecolor{lightgray}\hline
		s\_ilogb.c & ilogb(double) & 3654.7 & 8.3 & 16.7 & 75.0 & 438.7 & 58.3 \\\arrayrulecolor{lightgray}\hline
		s\_log1p.c & log1p(double) & 11913.7 & 9.9 & 61.1 & 88.9 & 1205.7 & 27.8 \\\arrayrulecolor{lightgray}\hline
		s\_logb.c & logb(double) & 1064.4 & 0.3 & 50.0 & 83.3 & 3131.8 & 33.3 \\\arrayrulecolor{lightgray}\hline
		s\_modf.c & modf(double, double*) & 1795.1 & 3.5 & 50.0 & 100.0 & 507.0 & 50.0 \\\arrayrulecolor{lightgray}\hline
		s\_nextafter.c & nextafter(double, double) & 7777.3 & 17.5 & 50.0 & 79.6 & 445.4 & 29.6 \\\arrayrulecolor{lightgray}\hline
		s\_rint.c & rint(double) & 5355.8 & 3.0 & 35.0 & 90.0 & 1808.3 & 55.0 \\\arrayrulecolor{lightgray}\hline
		s\_sin.c & sin (double) & 667.1 & 0.3 & 75.0 & 100.0 & 1951.4 & 25.0 \\\arrayrulecolor{lightgray}\hline
		s\_tan.c & tan(double) & 704.2 & 0.3 & 50.0 & 100.0 & 2701.9 & 50.0 \\\arrayrulecolor{lightgray}\hline
		s\_tanh.c & tanh(double) & 2805.5 & 0.7 & 33.3 & 100.0 & 4075.0 & 66.7 \\\arrayrulecolor{black}\midrule
		MEAN &  & 6058.4 & 6.9 & 42.8 & 90.8 &3868.0 & 48.9\\
		\arrayrulecolor{black}\bottomrule
	\end{tabular}
\end{table*}

Tab.~\ref{tab:eval:me_vs_austin} reports the testing results of Austin and CoverMe. It uses the same set of benchmarks as Tab.~\ref{tab:eval:me_vs_random} (Col.~1-2).  We use the time (Col.~3-4) and the branch coverage metric (Col. 5-6) to evaluate the efficiency and the coverage.  We choose branch coverage instead of the other three metrics in Gcov since branch coverage is the major concern in this paper. Besides, Gcov needs to have access to the generated test inputs to report the coverage, but currently, there is no viable way to access the test inputs generated by Austin.  Unlike Rand, the branch coverage percentage of Austin (Col.~5)  is provided by Austin itself, rather than by Gcov.

Austin also shows large performance variances over different benchmarks, from 667.1 seconds (the program {\tt sin}) to hours. As shown in the last row of Tab.~\ref{tab:eval:me_vs_austin}, Austin needs  6058.4 seconds on average for the testing. It should be mentioned that the average time does not include the benchmarks where Austin crashes or times out. Compared with Austin, CoverMe is  faster (Tab.~\ref{tab:eval:me_vs_austin}, Col.~4)  with 6.9 seconds on average (The results are also shown in Tab.~\ref{tab:eval:me_vs_random}(Col.~6)). 

CoverMe achieves a much higher branch coverage (90.8\%) than Austin (42.8\%).  We also compare across Tab.~\ref{tab:eval:me_vs_austin} and Tab.~\ref{tab:eval:me_vs_random}. On average, Austin provides slightly better branch coverage (42.8\%) than Rand (38.0\%). 

Col.~7-8 are the improvement metrics of CoverMe against Austin. The Speedup (Col.~7) is calculated as the ratio of the time spent by Austin and the time spent by CoverMe. The coverage improvement (Col.~7)  is calculated as the difference between the branch coverage of CoverMe and that of Austin.  We observe that CoverMe provides 3,868X speedup and  48.9\% coverage improvement on average. 


\begin{remark}
  Our evaluation shows that CoverMe has achieved high code coverage on
  most tested programs. One may wonder whether the generated inputs
  have triggered any latent bugs. Note that when no specifications are
  given, program crashes have been frequently used as an oracle for
  finding bugs in integer programs.  Floating-point programs, on the
  other hand, can silently produce wrong results without
  crashing. Thus, program crashes cannot be used as a simple, readily
  available oracle as for integer programs. Our experiments,
  therefore, have focused on assessing the efficiency of ME in solving
  the problem defined in Def.~\ref{def:overview:branchcoverage} or
  Lem.~\ref{def:overview:pb}, and do not evaluate its effectiveness in
  finding bugs, which is orthogonal and interesting future work.
\end{remark}


\section{Incompleteness}
\label{sect:discussion}


The development of $\ME$ toward a general solution to the search problem, as illustrated in this presentation, has been grounded in the
concept of representing function and mathematical optimization.  While
our theory guarantees that the $\ME$ procedure solves the search
problem correctly (Cor.~\ref{cor:theory:theoretical}), the practical
incompleteness (Remark~\ref{remark:theory:practical}) remains a challenge
in applying $\ME$ to real-world automated testing problems.  Below we
discuss three sources of practical incompleteness and how it can be mitigated. 

\Paragraph{Well-behaved Representing Functions.}
The \emph{well-behavedness} of the representing function is essential
for the working of $\ME$.  Well-behavedness is a common term in the MO
literature to describe that the MO backend under discussion can efficiently handle the objective function.  A common
requirement, for example, is that the objective function (or the
representing function for $\ME$) should be smooth to some degree.

Consider again the path reachability problem in
Sect.~\ref{sect:overview:path}.  If $\FOOI$ is instrumented as in
Fig.~\ref{fig:theory:discontinuity}(left), the corresponding \FOOR
also satisfies the conditions in Def.~\ref{def:theory:rp}, but this \FOOR is discontinuous at $\{-3,1,2\}$
(Fig.~\ref{fig:theory:discontinuity}(right)). The representing function
changes its value disruptively at its minimum points.  These minimum
points cannot be calculated by MO tools.
\begin{figure}
\begin{minipage}{0.5\linewidth}
\begin{lstlisting}[label={lst:eg1_b}]
void FOO_I(double x){
   @\highlight{r = r + (x <= 1)? 0:1; }@
@$l_0$@: if (x <= 1) x++;
   @\,@double y = square(x);
   @\highlight{r = r + (y == 4)? 0:1 ;}@
@$l_1$@: if (y == 4) ...;
}
\end{lstlisting}%
\end{minipage}
\begin{minipage}{0.49\linewidth}
\includegraphics[width=1\linewidth]{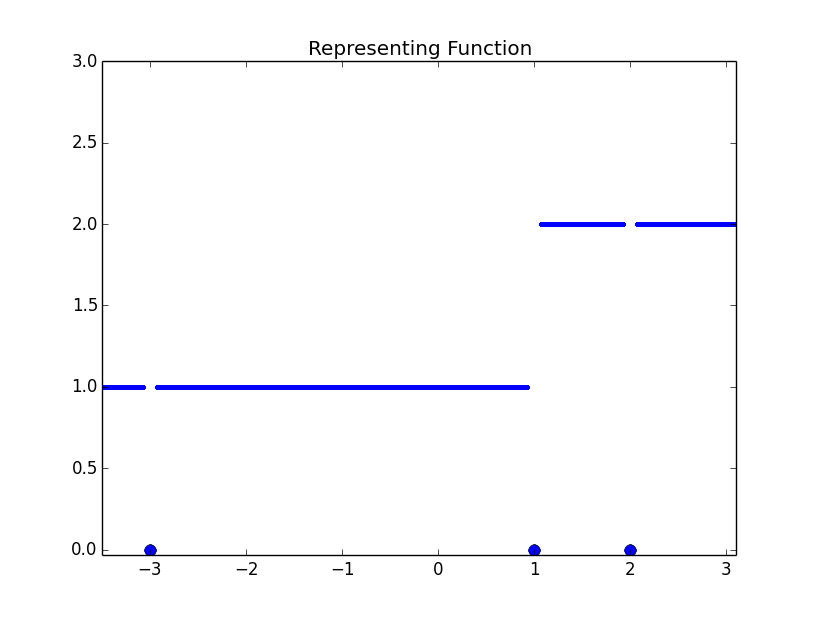}
\end{minipage}
\caption{An example of ill-behaved representing function for the path reachability problem in Sect.~\ref{sect:overview:path}. The left is the instrumented $\FOOI$; the right is the graph of representing function $\FOOR$.}
\label{fig:theory:discontinuity}
\end{figure}
The example also shows that the conditions in
Def.~\ref{def:theory:rp}(a-c) can be insufficient for avoiding
ill-behaved representing functions. However, it is unclear to us
whether there exists a set of conditions that are both easy-to-verify and
sufficient for excluding ill-behaved functions.

\Paragraph{Reliance on Program Execution.}
We have seen that the $\ME$ procedure generates test inputs of program
\FOO by minimizing, and therefore, running, another program
\FOOR. This execution-reliant feature has both benefits and risks.
 It allows us to generate test inputs of complex programs without
analyzing their semantics; it also means that ME can give different
results with different compilers or machines. 

In our experiments, we realize that the feature becomes 
disadvantageous if inaccuracy occurs in the program execution.
Consider the satisfiability testing problem with the constraint
$\pi=x \geq \num{1e-20}$.  Suppose we use the representing function
$\repf_{\pi}(x)= x\geq\num{1e-20}~?~0:(x-\num{1e-20})^2$ and implement
$\repf_{\pi}$ as a {\tt double}-precision floating-point program.  We can verify that $\repf_{\pi}$ is a representing function in the sense of Def.~\ref{def:theory:rp}, but it evaluates to $0$ not only when
$x\geq \num{1e-20}$, but also when $0\leq x<\num{1e-20}$ because the
smallest machine-representable {\tt double} is in the order of
$\num{1E-324}$~\cite{goldberg1991every}.  Then the $\ME$ procedure
may return $x^*=0$ in step M2 and then returns ``not found'' in step M3
because $x^*=0$ is not a model of $\pi$.

This issue described above may be mitigated using arbitrary-precision
arithmetic~\cite{Bailey:2005:HFA:1069590.1069679}. Another option, as
is demonstrated in the XSat solver~\cite{xsat}, is to construct the
representing function based on ULP, the \emph{unit in the last place}
value~\cite{goldberg1991every}, which avoids floating-point
inaccuracy by using the binary representation of floating-points.

\Paragraph{Efficient  MO Backend.}
The inherent intractability of global optimization is another source
of the incompleteness of $\ME$.  Even if the representing function is
well-behaved and free from computational inaccuracy, it is possible
that the MO backend returns sub-optimal or local minimum points due to
weak implementation, high-dimensional or intrinsically difficult
problems, \etc. That said, MO is still in active development, and our
$\ME$ approach has the flexibility to leverage its state-of-the-art,
such as the Basinhopping algorithm used in our experiments.


\section{Related Work}
\label{sect:relwork}
\Paragraph{Two Paradigms for Program Correctness.}
Constructing an axiomatic system~\cite{Hoare:1969:ABC:363235.363259,floyd1967assigning}
is  of central importance in ensuring program correctness. Let $\FOO$
be a program. If we write $\sem{\FOO}$ for the set of $\FOO$'s
possible execution paths (\aka trace semantics~\cite{DBLP:conf/esop/MauborgneR05}),
and $\err$ for the unsafe paths, then the correctness of \FOO is to ensure
\begin{align}
\sem{\FOO} \cap \err = \emptyset. 
\label{eq:relwork:corr}
\end{align}
 The problem is known to be
undecidable, and approximate solutions have been extensively
studied. One is {abstract interpretation}~\cite{DBLP:dblp_conf/popl/CousotC77,Cousot:1979:SDP:567752.567778}, which  systematically constructs
$\sem{\FOO}^{\sharp}\supseteq \sem{\FOO}$, and \emph{prove} Eq.~\eqref{eq:relwork:corr} by proving
$\sem{\FOO}^{\sharp}\cap\err=\emptyset$.~\footnote{The
  relationship between the abstract and concrete semantics is more
  commonly formalized with a {Galois
    connection}~\cite{DBLP:dblp_conf/popl/CousotC77}; we use
  $\sem{\FOO}^{\sharp}\supseteq \sem{\FOO}$ as a simplified case.}
Another category of approximation is {automated testing}. It
attempts to \emph{disprove} Eq.~\eqref{eq:relwork:corr} by generating inputs $x$ such that
$\sem{\FOO}(x) \in \err$, where $\sem{\FOO}(x)$ denotes the path
of executing $\FOO$ with input $x$.  

\Paragraph{Automated Testing.}
Many automated testing approaches adopt 
 \emph{symbolic execution} %
~\cite{King:1976:SEP:360248.360252,Boyer:1975:SFS:800027.808445,Clarke:1976:SGT:1313320.1313532,DBLP:journals/cacm/CadarS13}. 
They repeatedly select a target path $\tau$ and gather the conjunction of
logic conditions along the path, denoted by $\Phi_{\tau}$ (called path condition~\cite{King:1976:SEP:360248.360252}). They then use
SMT solvers to calculate a model of $\Phi_{\tau}$.  These approaches
are both sound and complete in the sense of 
$\FOO(x) \failure ~\Leftrightarrow~ x \models  \Phi_{\tau}$.
Symbolic execution and its variants have seen much progress since the
breakthrough of
SAT/SMT~\cite{Silva:1997:GNS:244522.244560,DBLP:journals/jacm/NieuwenhuisOT06},
but still have difficulties in handling numerical code. There are
two well-known efficiency issues. First, {path explosion} makes it
time-consuming to select a large number of $\tau$ and gather $\Phi_{\tau}$.  Second,
for numerical code, each gathered $\Phi_{\tau}$ involves numerical constraints that
can quickly go beyond the capabilities of modern SAT/SMT solvers.

A large number of search-based heuristics have been proposed to
mitigate issues from symbolic execution. They use fitness functions to
capture path conditions and use numerical optimization 
to minimize/maximize the fitness functions.  A well-developed
search-based testing is Austin~\cite{lakhotia2013austin}, which
combines symbolic execution and search-based heuristics. Austin has
been compared with Cute~\cite{Sen:2005:CCU:1081706.1081750}, a
dynamic symbolic execution tool, and shows to be significantly more
efficient.  These search-based solutions, however, as McMinn points
out in his survey paper~\cite{McMinn:2004:SST:1077276.1077279}, ``are
not standalone algorithms in themselves, but rather strategies ready
for adaption to specific problems.''

\Paragraph{Mathematical Optimization  in Testing and Verification.}
In the seminal work of Miller
\etal~\cite{Miller:1976:AGF:1313320.1313530}, optimization methods are
already used in generating test data for numerical code. These methods
are then taken up in the 1990s by
Koral~\cite{Korel:1990:AST:101747.101755,Ferguson:1996:CAS:226155.226158},
which have found their ways into many mature
implementations~\cite{Lakhotia:2010:FSF:1928028.1928039,Tillmann:2008:PWB:1792786.1792798,
  souza2011coral}. Mathematical optimization has also been employed in
program
verification~\cite{cousot2005proving,roozbehani2005convex}. Liberti
\etal~\cite{liberti2010mathematical,goubault2010static} have proposed
to calculate the invariant of an integer program as the mathematical
optimization problem of mixed integer nonlinear programming
(MINLP)~\cite{tawarmalani2004global}.  Recently, a floating-point
satisfiability solver XSat~\cite{xsat} has been developed. It
constructs a floating-point program $\repf_{\pi}$ from a formula $\pi$
in conjunctive normal form, and then decides the satisfiability of
$\pi$ by checking the sign of the minimum of $\repf_{\pi}$. This
decision procedure is an application of Thm.~\ref{thm:theory:rp}(a),
and it calculates the models of $\pi$ using
Thm.~\ref{thm:theory:rp}(b).  Compared with XSat, our work lays out
the theoretical foundation for a precise, systematic approach to
testing numerical code; XSat is an instance of our proposed $\ME$
procedure (see Example~\ref{eg:theory:sat}(b) and
Sect.~\ref{sect:overview:sat}).

\section{Conclusion}
 \label{sect:conc}


This paper introduces  Mathematical Execution (ME), a
new, unified approach for testing numerical code.  Our insight is to
(1) use a representing function to precisely and uniformly capture the
desired testing objective, and (2) use mathematical optimization to
direct input and program space exploration.  What $\ME$ --- and most
importantly, the accompanying representation function --- provides is
an approach that can handle a variety of automated testing
problems, including coverage-based testing, path reachability testing,
boundary value analysis, and satisfiability checking. We have
implemented a branch coverage testing tool as a proof-of-concept
demonstration of ME's potential. Evaluated on a collection
of Sun's math library (used in Java, JavaScript, Matlab, and Android),
our tool \coverme achieves substantially better coverage results
(near-optimal coverage on all tested programs) when compared to random
testing and Austin (a  coverage-based testing tool that combines symbolic execution and search-based strategies).



\bibliographystyle{abbrvnat}
\bibliography{main}

 \appendix
\section{Untested Programs in Fdlibm}
\label{sect:untested}
 
The programs from the freely distributed math library Fdlibm
5.3~\cite{fdlibm:web} are used as our benchmarks.
Tab.~\ref{tab:appendix:untested} lists all untested programs and
functions, and explains the reason why they are not selected. Three types
of the functions are excluded from our evaluation.  They are (1)
functions without any branch, (2) functions involving
non-floating-point input parameters, and (3) static C functions.

\setlength{\tabcolsep}{6pt}
\begin{table*}\footnotesize
\caption{Untested programs and functions 
in benchmark suite Fdlibm  and  the  corresponding explanations.}
\centering
\begin{tabular}{l  l  |  l}
\toprule
Program & Entry function &  Explanation \\\arrayrulecolor{black}\midrule
e\_gamma\_r.c & ieee754\_gamma\_r(double) & no branch \\\arrayrulecolor{lightgray}\hline
e\_gamma.c & ieee754\_gamma(double) & no branch \\\arrayrulecolor{lightgray}\hline
e\_j0.c & pzero(double) & static C function \\\arrayrulecolor{lightgray}\hline
 & qzero(double) & static C function \\\arrayrulecolor{lightgray}\hline
e\_j1.c & pone(double) & static C function \\\arrayrulecolor{lightgray}\hline
 & qone(double) & static C function \\\arrayrulecolor{lightgray}\hline
e\_jn.c & ieee754\_jn(int, double) & unsupported input type \\\arrayrulecolor{lightgray}\hline
 & ieee754\_yn(int, double) & unsupported input type \\\arrayrulecolor{lightgray}\hline
e\_lgamma\_r.c & sin\_pi(double) & static C function \\\arrayrulecolor{lightgray}\hline
 & ieee754\_lgammar\_r(double, int*) & unsupported input type \\\arrayrulecolor{lightgray}\hline
e\_lgamma.c & ieee754\_lgamma(double) & no branch \\\arrayrulecolor{lightgray}\hline
k\_rem\_pio2.c & kernel\_rem\_pio2(double*, double*, int, int, const int*) & unsupported input type \\\arrayrulecolor{lightgray}\hline
k\_sin.c & kernel\_sin(double, double, int) & unsupported input type \\\arrayrulecolor{lightgray}\hline
k\_standard.c & kernel\_standard(double, double, int) & unsupported input type \\\arrayrulecolor{lightgray}\hline
k\_tan.c & kernel\_tan(double, double, int) & unsupported input type \\\arrayrulecolor{lightgray}\hline
s\_copysign.c & copysign(double) & no branch \\\arrayrulecolor{lightgray}\hline
s\_fabs.c & fabs(double) & no branch \\\arrayrulecolor{lightgray}\hline
s\_finite.c & finite(double) & no branch \\\arrayrulecolor{lightgray}\hline
s\_frexp.c & frexp(double,  int*) & unsupported input type \\\arrayrulecolor{lightgray}\hline
s\_isnan.c & isnan(double) & no branch \\\arrayrulecolor{lightgray}\hline
s\_ldexp.c & ldexp(double, int) & unsupported input type \\\arrayrulecolor{lightgray}\hline
s\_lib\_version.c & lib\_versioin(double) & no branch \\\arrayrulecolor{lightgray}\hline
s\_matherr.c & matherr(struct exception*) & unsupported input type \\\arrayrulecolor{lightgray}\hline
s\_scalbn.c & scalbn(double, int) & unsupported input type \\\arrayrulecolor{lightgray}\hline
s\_signgam.c & signgam(double) & no branch \\\arrayrulecolor{lightgray}\hline
s\_significand.c & significand(double) & no branch \\\arrayrulecolor{lightgray}\hline
w\_acos.c & acos(double) & no branch \\\arrayrulecolor{lightgray}\hline
w\_acosh.c & acosh(double) & no branch \\\arrayrulecolor{lightgray}\hline
w\_asin.c & asin(double) & no branch \\\arrayrulecolor{lightgray}\hline
w\_atan2.c & atan2(double, double) & no branch \\\arrayrulecolor{lightgray}\hline
w\_atanh.c & atanh(double) & no branch \\\arrayrulecolor{lightgray}\hline
w\_cosh.c & cosh(double) & no branch \\\arrayrulecolor{lightgray}\hline
w\_exp.c & exp(double) & no branch \\\arrayrulecolor{lightgray}\hline
w\_fmod.c & fmod(double, double) & no branch \\\arrayrulecolor{lightgray}\hline
w\_gamma\_r.c & gamma\_r(double, int*) & no branch \\\arrayrulecolor{lightgray}\hline
w\_gamma.c & gamma(double, int*) & no branch \\\arrayrulecolor{lightgray}\hline
w\_hypot.c & hypot(double, double) & no branch \\\arrayrulecolor{lightgray}\hline
w\_j0.c & j0(double) & no branch \\\arrayrulecolor{lightgray}\hline
 & y0(double) & no branch \\\arrayrulecolor{lightgray}\hline
w\_j1.c & j1(double) & no branch \\\arrayrulecolor{lightgray}\hline
 & y1(double) & no branch \\\arrayrulecolor{lightgray}\hline
w\_jn.c & jn(double) & no branch \\\arrayrulecolor{lightgray}\hline
 & yn(double) & no branch \\\arrayrulecolor{lightgray}\hline
w\_lgamma\_r.c & lgamma\_r(double, int*) & no branch \\\arrayrulecolor{lightgray}\hline
w\_lgamma.c & lgamma(double) & no branch \\\arrayrulecolor{lightgray}\hline
w\_log.c & log(double) & no branch \\\arrayrulecolor{lightgray}\hline
w\_log10.c & log10(double) & no branch \\\arrayrulecolor{lightgray}\hline
w\_pow.c & pow(double, double) & no branch \\\arrayrulecolor{lightgray}\hline
w\_remainder.c & remainder(double, double) & no branch \\\arrayrulecolor{lightgray}\hline
w\_scalb.c & scalb(double, double) & no branch \\\arrayrulecolor{lightgray}\hline
w\_sinh.c & sinh(double) & no branch \\\arrayrulecolor{lightgray}\hline
w\_sqrt.c & sqrt(double) & no branch \\\arrayrulecolor{black}\bottomrule
\end{tabular}
\label{tab:appendix:untested}
\end{table*}

\section{Implementation Details}
\label{sect:implem}










\begin{figure}
\centering
\includegraphics[width=1.0\linewidth]{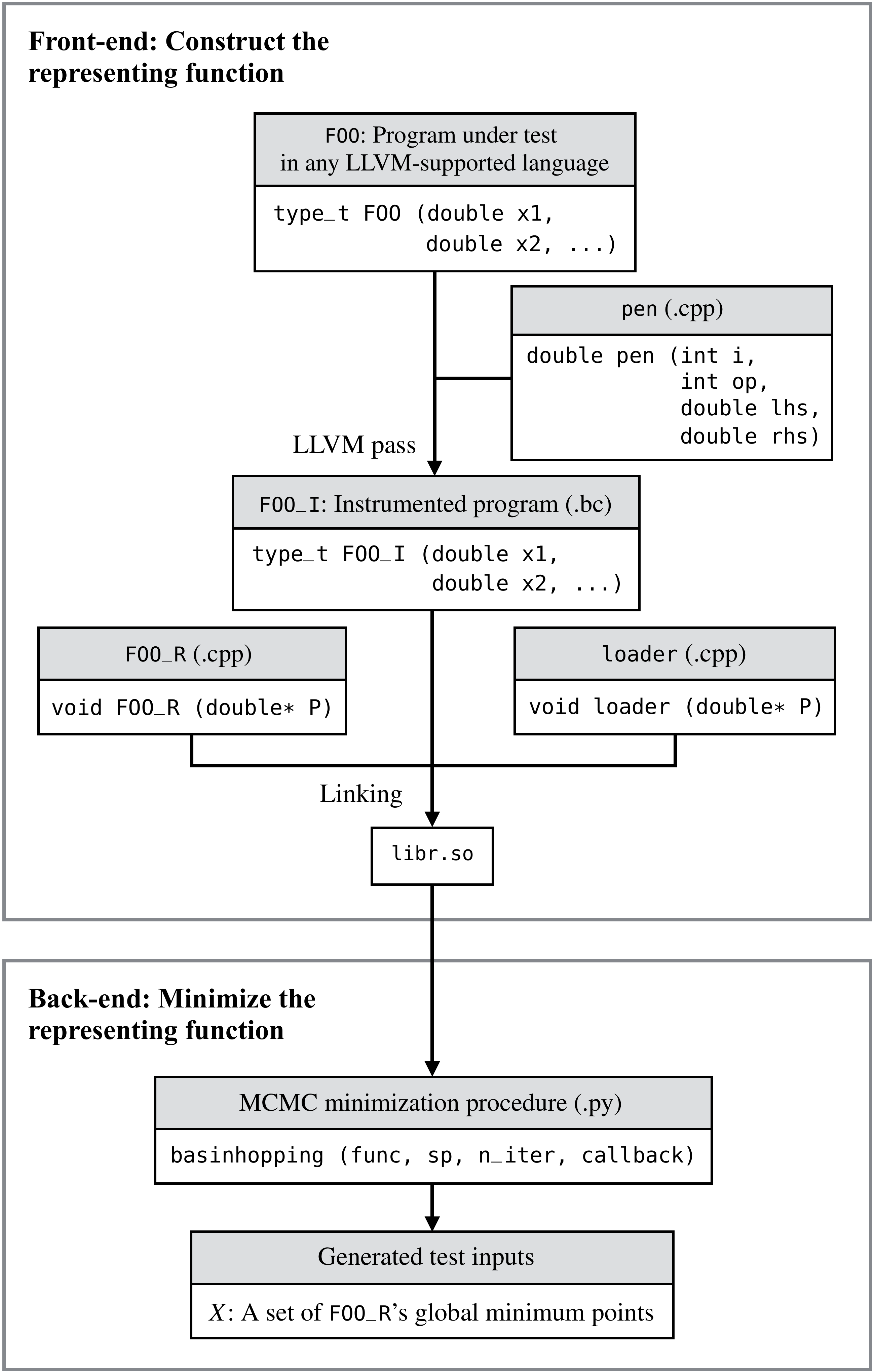}
\caption{\coverme Implementation.}
\label{fig:implem:overview}
\end{figure}

As a proof-of-concept demonstration, we have implemented
Algo.~\ref{theory:algo:coverme} in the tool \coverme.  This section
presents the implementation and technical details omitted from the main body of the paper.





\subsection{Frontend of CoverMe}
The frontend implements Step 1 and Step 2 of Algo.~\ref{theory:algo:coverme}. 
\coverme compiles the program under test \FOO to LLVM IR with Clang~\cite{clang:web}. Then it uses an LLVM pass~\cite{llvmpass:web} to inject assignments in \FOO. The program under test can be in any LLVM-supported language, \eg, Ada, the C/C++ language family, or Julia. Our current implementation accepts C code only. 


Fig.~\ref{fig:implem:overview} illustrates \FOO as a function of signature \texttt{type\_t
  FOO (type\_t1 x1, type\_t2 x2, \ldots)}.  The return type (output) of the function,
\texttt{type\_t}, can be any kind of types supported by C, whereas the types of the input
parameters, {\texttt{type\_t1, type\_t2, \ldots}}, are restricted to
{\tt double} or {\tt
  double*}. 
We have explained the signature of $\pen$ in Def.~\ref{def:overview:pen}.
Note that \coverme does not inject $\pen$ itself into \FOO, but instead injects
assignments that invoke $\pen$.  We implement $\pen$ in a separate C++
file.



The frontend also links \FOOI and \FOOR with  a simple program  \loader to generate a shared object file \librso, which is the 
outcome of the frontend. It stores the
representing function in the form of a shared object file (.so file). 

 
\subsection{Backend of CoverMe}

The backend implements Step 3 of Algo.~\ref{theory:algo:coverme}. 
It invokes the representing function via \librso.  
The kernel of the backend is an external MCMC engine.
It uses the off-the-shelf implementation known as \emph{Basinhopping} from the Scipy Optimization package~\cite{scipy:web}. Basinhopping takes a range of input parameters. Fig.~\ref{fig:implem:overview} shows the important ones for our implementation {\tt
  basinhopping(f, sp, n\_iter, call\_back)}, where {\tt f} refers to the
representing function  from \librso, {\tt sp} is a starting point as a 
Python {\tt Numpy} array, {\tt n\_iter} is the iteration number used in Algo.~\ref{theory:algo:coverme} and {\tt call\_back} is a
client-defined procedure. Basinhopping invokes {\tt call\_back}  at the end of each
iteration (Algo.~\ref{theory:algo:coverme}, Lines~24-34). The call\_back procedure
allows \coverme to terminate if it saturates all branches. In this
way, \coverme does not need to wait until passing all $\nStart$
iterations (Algo.~\ref{theory:algo:coverme}, Lines~8-12).



\subsection{Technical Details}
Sect.~\ref{sect:overview} assumes Def.~\ref{def:overview:branchcoverage}(a-c) for the sake of simplification.   This section discusses how   \coverme relaxes the assumptions when  handling real-world floating-point code.  We also show how \coverme handles function calls at the end of this section. 

\Paragraph{Dealing with Pointers  (Relaxing Def.~\ref{def:overview:branchcoverage}(a))} 
We consider only pointers to floating-point numbers. They may occur  (1) in an input parameter, (2) in a conditional statement, or  (3) in the code body but not in the conditional statement. 

\coverme inherently handles case (3) because it is execution-based and does not need to analyze pointers and their effects.  \coverme currently does not handle case (2) and  simply ignores these conditional statements by not injecting $\pen$ before them.  

Below we explain how \coverme deals with case (1). A branch
coverage testing problem for a program whose inputs are pointers to
doubles, can be regarded as the same problem with a simplified program
under test. For instance, finding test inputs to cover branches of
program {\tt void FOO(double* p) \{if (*p <= 1)... \} } can be reduced
to testing the program {\tt void FOO\_with\_no\_pointer (double x)
  \{if (x <= 1)... \}}. \coverme transforms program \FOO to {\tt FOO\_with\_no\_pointer}
if a   \FOO's input parameter is a floating-point pointer.   




\Paragraph{Dealing with  Comparison between Non-floating-point Expressions (Relaxing Def.~\ref{def:overview:branchcoverage}(b)).} 

We have encountered situations where a  conditional statement invokes a comparison between non floating-point numbers.  \coverme handles these situations by first promoting the non floating-point numbers to floating-point numbers and then injecting $\pen$ as described in Algo.~\ref{theory:algo:coverme}.  For example, before a conditional statement like {\tt if (xi op yi)} where {\tt xi} and {\tt yi} are integers, \coverme injects  {\tt r = pen (i, op, (double) x, (double) y));}. 
Note that such an approach does not allow us to handle  data types that are incompatible with floating-point types, \eg, conditions like {\tt if (p != Null)}, which \coverme has to ignore.

\Paragraph{Dealing with Infeasible Branches (Relaxing Def.~\ref{def:overview:branchcoverage}(c)).}
Infeasible branches are branches that cannot be covered by any input.
Attempts to cover infeasible branches are useless and time-consuming. 

Detecting infeasible branches is a difficult problem in general.
\coverme uses a simple heuristic to detect and ignore infeasible
branches.  When CoverMe finds a minimum that is not zero, that is,
$\FOOR(x^*) > 0$, CoverMe deems the unvisited branch of the last
conditional to be infeasible and adds it to $\Explored$, the set
of unvisited and deemed-to-be infeasible branches.

 Imagine that we modify $l_1$ of the program \FOO in
Fig.~\ref{fig:algo:injecting_pen} to  the conditional statement {\tt
  if (y == -1)}. Then the branch $1_T$ becomes infeasible. We rewrite this
modified program below and  illustrate how we deal with infeasible branches. 
\begin{lstlisting}
l0: if (x <= 1) {x++}; 
    y = square(x); 
l1: if (y == -1) {...}
\end{lstlisting}
where we omit the concrete implementation of {\tt square}. 

Let \FOOR denote the representing function constructed for the
program. In the minimization process, whenever \coverme obtains $x^*$
such that $\FOOR(x^*) > 0$, \coverme selects a branch that it regards
infeasible. \coverme selects the branch as follows: Suppose $x^*$
exercises a path $\tau$ whose last conditional statement is denoted by
$l_z$, and, without loss of generality, suppose $z_T$ is passed
through by $\tau$, then \coverme regards $z_F$ as an infeasible
branch.

In the modified program above, if $1_F$ has been saturated, the
representing function evaluates to $(y+1)^2$ or $(y+1)^2+1$, where $y$
equals to the non-negative {\tt square(x)}. Thus, the minimum point
$x^*$ must satisfy $\FOOR(x^*)>0$ and its triggered path ends with
branch $1_F$.  \coverme then regards $1_T$ as an infeasible branch.

\coverme then regards the infeasible branches as already saturated. It
means, in line 12 of Algo.~\ref{theory:algo:coverme}, \coverme updates
$\Explored$ with saturated branches and infeasible branches (more
precisely, branches that \coverme regards infeasible).

The presented heuristic works well in practice (See
Sect.~\ref{sect:eval}), but we do not claim that our heuristic always
correctly detects infeasible branches.


\Paragraph{Dealing with Function Calls.}
By default, \coverme injects $\myr = \pen_i$ only in the entry
function to test. If the entry function invokes other external
functions, they will not be transformed. For example, in the program
\FOO of Fig.~\ref{fig:algo:injecting_pen}, we do not transform {\tt
  square(x)}. In this way, \coverme only attempts to saturate all
branches for a single function at a time.

However, \coverme can also easily handle functions invoked by its entry
function. As a simple example, consider:
\begin{lstlisting}
void FOO(double x) {GOO(x);}
void GOO(double x) {if (sin(x) <= 0.99) ... }
\end{lstlisting}

If \coverme aims to saturate {\tt FOO} and {\tt GOO} but not {\tt sin},
and  it sets \FOO as the entry function, then it instruments
both {\tt FOO} and {\tt GOO}. Only {\tt GOO} has a
conditional statement, and \coverme injects an assignment on {\tt
  r} in {\tt GOO}.

\end{document}